\documentclass[12pt,draftclsnofoot,onecolumn]{IEEEtran}
\usepackage{amsmath}
\usepackage{amssymb}
\usepackage{graphicx}
\usepackage{epstopdf}
\usepackage{graphics}
\usepackage{multicol}
\usepackage{amsmath,epsfig}
\usepackage[ruled,vlined]{algorithm2e}

\usepackage{pdfpages}
\usepackage{amsfonts}
\usepackage{amssymb}
\usepackage{amsthm}
\usepackage{bbm}
\usepackage{color}
\usepackage{cite}
\usepackage{float}
\usepackage{enumerate}

\newtheorem{corollary}{Corollary}

\theoremstyle{plain}
\newtheorem{thm}{\protect\theoremname}
\theoremstyle{plain}
\newtheorem{lem}[thm]{\protect\lemmaname}
\theoremstyle{plain}

\usepackage{mathtools}

\providecommand{\lemmaname}{Lemma}
\providecommand{\propositionname}{Proposition}
\usepackage{mathtools}
\providecommand{\theoremname}{Theorem}
\providecommand{\lemmaname}{Lemma}
\providecommand{\propositionname}{Proposition}
\providecommand{\theoremname}{Theorem}
\begin{document}
	\title{Stochastic Geometry Analysis of IRS-Assisted Downlink Cellular Networks}
	\author{Taniya~Shafique, Hina Tabassum, and Ekram Hossain\thanks{T. Shafique and E. Hossain are with the Department of Electrical and Computer Engineering, University of Manitoba, Canada. H. Tabassum is with the Department of Electrical Engineering and Computer Science, York University, Canada.} 
	}
	\maketitle
	\begin{abstract}
	Using stochastic geometry tools, we develop a comprehensive framework to analyze the downlink coverage probability, ergodic capacity, and energy efficiency (EE) of various types of users (e.g., users served by direct base station (BS) transmissions and indirect intelligent reflecting surface (IRS)-assisted transmissions) in a cellular network with multiple BSs and IRSs. The proposed stochastic geometry framework can capture the impact of channel fading, locations of BSs and IRSs, arbitrary phase-shifts  and  interference experienced by a typical user supported by direct transmission and/or IRS-assisted transmission. For IRS-assisted transmissions, we first model the desired signal power from the nearest IRS as a sum of scaled generalized gamma (GG) random variables whose parameters are functions of the IRS phase shifts. Then, we derive the Laplace Transform (LT)  of the received signal power in a closed form. Also, we model the aggregate interference from  multiple IRSs as the sum of normal random variables. Then, we derive the  LT of the aggregate interference from all IRSs and BSs. The derived LT expressions are used to calculate coverage probability, ergodic capacity, and EE for  users served by direct BS transmissions as well as  users served by IRS-assisted transmissions. Finally, we derive the overall network coverage probability, ergodic capacity, and EE based on the fraction of direct and  IRS-assisted users, which is defined as a function of the deployment intensity of IRSs, as well as blockage probability of direct transmission links. Numerical results validate the derived analytical expressions and extract useful insights related to the number of IRS elements, large-scale deployment of IRSs and BSs, and the impact of IRS interference on direct transmissions. 
	\end{abstract}
	
	
	\section{Introduction}
	\label{sec:Intro}
Intelligent reflecting surfaces (IRSs) are  considered as a key enabling technology for the sixth generation (6G) wireless communications systems. IRSs enable a smart manipulation of the wireless propagation environment \cite{wu2019towards,wu2021intelligent}. Each IRS consists of many antenna elements (a.k.a IRS elements) \cite{huang2019reconfigurable}   and each IRS element is controlled via a controller that assists each  IRS element to steer the incident  signal into the desired direction \cite{liaskos2018new}.
Also, with the advances in the wireless technology and extravagant demand of higher data rate to millions of indoor/outdoor devices, it has become inevitable to utilize the resources wisely to enable massive connectivity. In this context,  IRSs operate as a low cost solution to extend the communication range and to provide service to more users. In order to achieve this goal, the transmissions can happen  in three modes, i.e., (i)~\textit{Joint Transmission:} in which a user receives the IRS signals  combined with the direct signal from the base-stations (BSs),  (ii)~\textit{IRS-only Transmission:} in which a user receives only IRS transmissions and the direct transmissions get blocked, and (iii)~\textit{Direct Transmission:} in which a user gets served only through direct transmissions. 

It is noteworthy that combining the signals coming from the direct and indirect IRS-assisted path may suffer from  incoherent multi-path delays and it may necessitate sophisticated synchronization, detection, and co-phasing techniques resulting in complex hardware/software design. Furthermore, the impact of IRS transmissions is generally more understandable in the absence of direct link; therefore, it is crucial to investigate the significance of  IRS-only transmissions without direct links. Similarly, the fact that the direct transmissions from BSs may be impacted by the presence of IRSs, it is important to study the performance of direct transmissions in a large-scale IRS-assisted network comprehensively. 
In this paper, we develop a novel framework to analyze the performance of various types of transmissions in a   multi-BS, multi-IRS network. 


\subsection{Background Work}
To date, there have been a number of research works that considered the performance analysis of IRS-assisted communication systems assuming a single IRS, single source and  destination \cite{ basar2019wireless,kudathanthirige2020performance, yang2020accurate, boulogeorgos2020performance, trigui2020comprehensive, 9316943,tao2020performance, van2021outage}.
For instance,  the authors in \cite{basar2019wireless} applied the central limit theorem (CLT) to derive the approximate symbol error  probability expressions under independent Rayleigh fading channels.   In \cite{kudathanthirige2020performance}, the authors derived the approximate outage  probability, symbol error rate, and upper and lower bounds on ergodic capacity by applying CLT and assuming uncorrelated Rayleigh fading channels. Later, \cite{yang2020accurate} derived the average bit error rate, capacity, and outage probability with  Rayleigh fading. In \cite{boulogeorgos2020performance}, the exact outage probability, symbol error rate, and ergodic capacity expressions  under Rayleigh fading  were derived. In  \cite{trigui2020comprehensive} and \cite{9316943}, using moment generating function (MGF)-based approach, the exact outage probability was derived considering \textit{Nakagami-m} and generalized fading channels, respectively. The direct link transmission was ignored in \cite{basar2019wireless,kudathanthirige2020performance,trigui2020comprehensive, yang2020accurate}; however \cite{9316943} considered both the direct and IRS-assisted transmissions. In \cite{tao2020performance}, the authors derived the outage probability and ergodic capacity expressions considering IRS  link  modeled with Rician fading and direct channel modeled as Rayleigh fading. The joint direct and IRS-assisted transmission was considered.

\textbf{All of the aforementioned research works were limited to single IRS, single source, and single destination under varying fading channels. That is, the impact of interference is ignored and the IRS is deployed at a fixed location. Furthermore, the analyses assume optimal phase-shifts and apply CLT, which simplifies the cascaded signal model substantially}. Recently, in \cite{van2021outage}, using moment matching method, the authors derived the outage probability and capacity expressions under correlated Rayleigh fading channel while considering arbitrary phase shifts.  This work considered both the direct and IRS-assisted links; however, again the framework was limited to a single IRS, single source, and single destination. 
Another work is \cite{peng2021analysis} where multi-pair D2D network is considered with a single IRS and the authors derive average achievable rate expressions assuming arbitrary phase-shifts. However, the authors considered approximating the signal and interference power with their respective statistical averages.

Another series of research works considered  multiple IRSs, single source, and single destination \cite{yang2020outage,galappaththige2020performance, do2021multi}. In \cite{yang2020outage},  the authors derived the outage probability  considering Rayleigh fading with the direct transmission blocked.  The transmission is conducted by only one IRS that provides the maximum SNR. Instead of applying CLT, the authors proposed Generalized-$K$ approximation. In \cite{galappaththige2020performance}, the authors applied CLT to derive the outage probability and rate considering Nakagami-$m$ fading. Both the direct and indirect transmissions were considered. Similarly, in \cite{do2021multi}, the authors  derived the outage probability  by approximating the end-to-end IRS-assisted channel with the log-normal and gamma distribution. 

\textbf{The aforementioned research works ignored interference from IRSs, assumed a single BS, and ideal phase shifts were assumed. }
Recently, a couple of research works considered a realistic multi-IRS set-up with multiple BSs  \cite{zhu2020stochastic}. In \cite{zhu2020stochastic}, the authors derived the average achievable rate of the IRS-assisted multi-BS network and derived  the Laplace Transform (LT) of the aggregate interference from all BSs and IRSs. However, the  interference from all BSs to a specific  IRS is replaced by its average value. The resulting rate expressions require four-fold integral evaluations.
Another relevant research work is  \cite{lyu2020hybrid} where  the authors derived the coverage probability expressions considering joint direct and IRS transmission. The signal power is approximated with the Gamma distribution and approximated  the interference from all IRS with the  mean IRS interference.  
Both of the aforementioned works \cite{zhu2020stochastic,lyu2020hybrid} assumed optimal phase-shifts in the desired signal and interference which makes the application of CLT possible. 

\subsection{Paper Contribution and Organization}


In this paper, we develop a comprehensive framework to analyze the coverage probability and rate of various types of users (e.g., users performing direct transmissions and indirect IRS-assisted transmissions) in a realistic large-scale multi-BS, multi-IRS network. The proposed framework can capture the impact of arbitrary phase-shifts on the received signal power as well as the aggregate interference from all IRSs on users that are served by direct transmissions from BS or IRS-assisted transmissions. More specifically, we have the following main contributions:

\begin{itemize}
\item For IRS-assisted downlink transmissions, we  characterize the desired signal power from the nearest IRS as a sum of scaled generalized gamma (GG) random variables whose parameters are a function of the IRS phase shifts. Then, we derive the novel LT expression and validate its accuracy considering both the optimized and randomized  phase-shifts of the IRS. 
\item We characterize the aggregate interference from  multiple IRSs in a multiple BS scenario as the sum of normal random variables. Then, we derive the  LT of the aggregate interference from all IRSs. The derived expressions can be customized for both types of users, i.e., those served by direct BS transmissions and those served by IRS-assisted transmissions.
\item Based on the LT expressions, we characterize the coverage probability, ergodic capacity, and energy-efficiency of both the IRS-assisted users and direct users.
\item Finally, we derive the overall coverage probability, ergodic capacity, and energy efficiency based on the fraction of direct and indirect IRS-assisted users in the network. This fraction is derived as a function of the \textbf{(i) }deployment intensity of IRSs as well as  \textbf{(ii)} blockage probability of direct transmission links.
\item The analytical results are validated by Monte-Carlo simulations. Numerical results extract useful insights related to the impact of IRS interference on IRS-assisted as well as direct transmissions in a large-scale network as a function of the number of IRS elements, intensity of IRSs and BSs, and the transmit power of BSs.
\end{itemize}  


The remainder of the paper is organized as follows. Section~II describes the  system model and assumptions and the methodology of analysis. We characterize the statistics of the received signal power, aggregate interference and  the corresponding LT of users supported by IRS in Section~III and Section~IV, respectively. The coverage probability, ergodic capacity and energy efficiency of users supported by IRS transmissions and users supported by direct transmission, also the overall coverage of the network and achievable data rate is provided in  Section~V. Then, in Section~VI, we present selected numerical results  followed by conclusions in Section~VII. A list of the major notations is presented in {\bf Table~1}.
	\section{System Model and Assumptions}
	\label{sec:SystemModel}
In this section, we present the network, transmission,  signal and interference models for users who are served by direct BS transmissions and those served by IRS-assisted transmissions, and also our methodology for large-scale analysis of the system.

	\subsection{Network Deployment  and Transmission Model}
	We consider a two-tier downlink cellular network consisting of IRS surfaces, BSs, and users within a coverage area of radius $R$. The locations of the BSs follow a homogeneous Poisson Point Process (PPP) denoted as  $\Phi_{\rm B}$ with intensity $\lambda_{B}$, whereas the locations of the IRSs follow  Binomial Point Process (BPP) in which $M$ IRSs are distributed uniformly in the coverage region. {For simplicity, we refer to $\lambda_R=\frac{M}{\pi R^2}$ as the IRS intensity throughout the paper}. We assume that  the IRSs are deployed at a fixed height $H_{\rm R}$ and are equipped with $N$ elements each, whereas all the BSs have a fixed height $H_B$.  We assume that there are two different types of users in the considered multi-BS and multi-IRS network, i.e.,
	\begin{itemize}
	    \item {\em Direct users:} who are served by direct BS transmissions, and 
	    \item {\em IRS-assisted users:} who are served by indirect IRS-assisted  transmissions. 
	\end{itemize}
	The typical user who is deployed at origin would reflect the performance of any user within the coverage region.  We also consider $\mathcal{A}$  IRS-assisted users and $1-\mathcal{A}$ direct users in the system. For direct transmission from the BS, the typical user is associated to the nearest BS. In the indirect IRS-assisted transmission mode, the user associates to the nearest IRS, and then, that  nearest IRS associates to the nearest BS (as illustrated in Fig.~1). 	
	
	We assume that an IRS can relay information from only one BS to only one user at a predefined time/frequency resource to maintain orthogonality. 	We consider that the \textit{direct} communication  (i.e., BS to the typical user) and \textit{indirect} IRS-assisted communication  (i.e., BS to IRS and IRS to the typical user) share different frequency spectrum such that a BS can serve both the direct and indirect IRS-assisted users. 
	

		
	\begin{figure}
		\begin{center}			\includegraphics[scale=0.85]{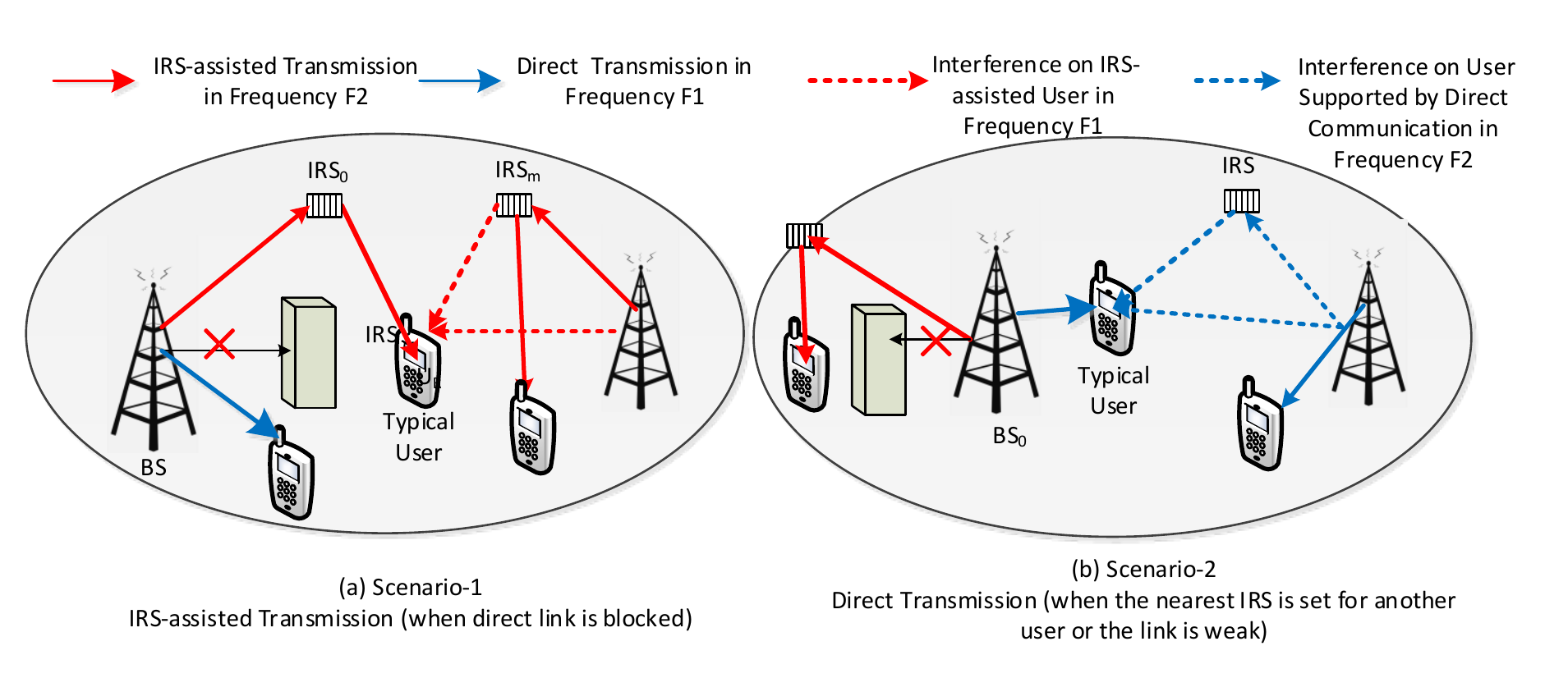}
			\caption{System model for direct and IRS-assisted  communication in multi-IRS and multi-BS setup: (a) {Scenario 1: when the user is connected with a BS through nearest IRS and the direct link to the nearest BS is blocked,} and (b) Scenario 2: when the user is connected with a nearest BS in the presence a weak IRS link.  }			\label{fig:Tania2}
		\end{center}
	\end{figure}

\begin{table*}[!h]
\centering
\small
\caption{{\color{black} Mathematical notations} }
\resizebox{\textwidth}{!}{\begin{tabular}{|c|c|c|c|}
\hline
  \textbf{Notation}&\textbf{Description} & \textbf{Notation}&\textbf{Description}\\ \hline
   \footnotesize	$R$  &\footnotesize Coverage radius  & \footnotesize $N$&\footnotesize Number of IRS elements\\ \hline
   \footnotesize	 $M$&\footnotesize Total number of IRSs  & \footnotesize $\lambda_{\rm B}$&\footnotesize BS intensity\\ \hline
    \footnotesize     $\Phi_{\rm B}$&\footnotesize PPP for BSs&  $ H_{\rm B}$,$H_{\rm R}$ &  \footnotesize BS height, IRS height\\ \hline
  \footnotesize $\alpha$ &\footnotesize  Path-loss exponent & $BS_{\rm j}$& \footnotesize   j-th interfering BS to the typical user in  direct mode \\ \hline  $BS_{\rm 0}$& \footnotesize   Nearest BS to the typical user in  direct mode & $\ell_{\rm j}$ &\footnotesize  Direct distance of  $BS_{\rm j}$ in 2D\\ \hline
   $d_{\rm j}$ &\footnotesize  Direct distance of  $BS_{\rm j}$ in 3D & $\ell_{\rm 0}$ &\footnotesize  Direct distance of nearest $BS_{\rm 0}$ in 2D\\ \hline
  \footnotesize $d_{\rm 0}$ &\footnotesize  Direct distance of  $BS_{\rm j}$ in 3D & $h_{\rm j}$ & \footnotesize  Fading channel between typical user and  $BS_{\rm j}$   \\ \hline
  \footnotesize $h_{\rm 0}$ &\footnotesize  Fading channel between typical user and $BS_{\rm 0}$  & $\hat{P},P$ &\footnotesize  Transmission power of BSs in direct and IRS-assisted mode \\ \hline
  \footnotesize IRS$_{\rm m_n}$ & \footnotesize  The $n$-th element of interfering IRS $m$ & IRS$_{\rm 0_n}$ & The $n$-th element of nearest IRS$_0$ to  user \\ \hline$\vert{ g}_{0,{m_n}}\vert ,{\phi_{0,m_n}} $ &\footnotesize Fading magnitude and phase  from the user to the IRS$_{\rm m_n}$ & $r_{0,m_n} $ &\footnotesize Distance of the user to the IRS$_{\rm m_n}$    \\ \hline
  \footnotesize $ \vert{ f}_{{m_n,j}}\vert ,{\psi_{m_n,j}}$ &\footnotesize Fading magnitude and phase from IRS$_{\rm m_n}$ to $BS_{\rm j}$ & ${t_{m_n,j}} $ &\footnotesize  Distance from IRS$_{\rm m_n}$ to $BS_{\rm j}$  \\ \hline
  \footnotesize $ \Theta_{m_n}$ &\footnotesize   Phase shift of  IRS$_{\rm m_n}$  & 
  \footnotesize $ \Theta_{0_n}$ &\footnotesize   Phase shift of  IRS$_{\rm 0_n}$   \\ \hline$\vert{ g}_{0,{0_n}}\vert ,{\phi_{0,0_n}} $ &\footnotesize Fading magnitude and phase  from the user to the IRS$_{\rm 0_n}$ & $r_{0,0_n} $ &\footnotesize Distance of the user to the IRS$_{\rm 0_n}$    \\ \hline
  \footnotesize $ \vert{ f}_{{0_n,j}}\vert ,{\psi_{0_n,j}}$ &\footnotesize Fading magnitude and phase from IRS$_{\rm 0_n}$ to $BS_{\rm j}$ & ${t_{0_n,j}} $ &\footnotesize  Distance from IRS$_{\rm 0_n}$ to $BS_{\rm j}$  \\ \hline
  \footnotesize $I_B $ &\footnotesize  Aggregate interference from all BSs in IRS-assisted mode  & $I_R $ &\footnotesize Aggregate interference from all IRS in IRS-assisted mode    \\ \hline  \footnotesize $\hat{I}_B $ &\footnotesize  Aggregate interference from all BSs in direct mode    & $\hat{I}_R $ &\footnotesize  Aggregate interference from all IRS in direct mode  \\ \hline
  \footnotesize $S_{D_0} $ &\footnotesize 
  Received signal power from BS$_0$  & $S_{R_0} $ &\footnotesize  
  Received signal power from IRS$_0$  \\ \hline
  \footnotesize $ C_D$ &\footnotesize Coverage probability in the direct transmission mode   & $ C_{ID}$ &\footnotesize Coverage probability in the IRS-assisted mode   \\ \hline
  \footnotesize $ \gamma_D$ &\footnotesize Direct mode SINR  & $ \gamma_{ID} $ &\footnotesize Direct mode SINR   \\ \hline
  \footnotesize $\tau $ &\footnotesize SINR threshold   & $\mathcal{A} $ &\footnotesize Fraction of users assisted by IRS  \\ \hline
  \footnotesize $X_{G}(\kappa,\zeta )$ &\footnotesize Gamma RV and parameters   & $ X_{GG}( a,d,p)$ &\footnotesize Generalized gamma RV  and parameters \\ \hline
  \footnotesize $ R_D$ &\footnotesize Rate achieved in direct mode   & $ R_{ID}$ &\footnotesize Rate achieved in indirect mode  \\ \hline
  \footnotesize $t_{j}$ &\footnotesize Approximation of $t_{m,j}$   & $N_0$ &\footnotesize Noise power spectral density\\ \hline
  \footnotesize $p_{\rm BS},p_{\rm U}$ &\footnotesize Static power consumption of BS and user   & $p_{\rm  IRS} $ &\footnotesize IRS power consumption \\ \hline
  \footnotesize $p_{\rm ID}$ &\footnotesize Total system power consumption per user in ID mode & $p_{\rm D}$ &\footnotesize Total system power consumption per user in direct mode \\ \hline 
  \footnotesize $\mathcal{A}$ &\footnotesize Fraction of user associated with IRS-assisted ID communication  & $\beta $ &\footnotesize    Reference  channel power gain on  free space path loss at 1-meter distance\\\hline
\end{tabular}}
\label{Notation_Summary_mmwave}
\end{table*}   

\subsection{Signal and Interference Models (IRS-Assisted Users)}
\subsubsection{Desired Signal Power}
The signal power received at the typical user from the nearest IRS (IRS$_0$) is given as \cite{bjornson2020power,peng2021analysis}: 
\begin{align}
\label{eq:desiredInDirect1}
\begin{split}
S_{ R_0}= &P\;\vert\mathbf{ \hat g}^{H}_{0,0}  {\Theta_{0}} \;\mathbf{\hat f}_{0,{j}} \vert^2
= P \vert\sum_{n=1}^{N} C_{0_n,j} f_{0_n,j}  g_{0,0_n} {e^{j\theta_{0_n}} }\vert^2, 
\end{split}
\end{align}
	where $P$ is the transmission power of the  BSs in IRS-assisted mode,  $g_{0,0_n}=\vert{ g}_{0,{0_n}}\vert e^{-j\phi_{0,0_n}}$ is the Rayleigh fading channel gain from the typical user to the $n$-th element of IRS$_0$, thus ${\hat g}_{0,0_n}=\beta \left({r_{0,0_n}}\right)^{-\alpha/2} {g}_{0,0_n}$, where $\alpha\ge 2$ represents the path-loss exponent,  $\beta=\left(\frac{4 \pi f_c}{c}\right)^{-2}$ is  the  channel power gain on  free-space 
	path-loss model at a reference distance of one meter, $f_c$ is carrier frequency, and $c$ represents the speed of light,  and 
${{\bf \hat  g}}_{0,0} \in \mathbb{C}^{1\times N}$, where  $r_{0,0_n}=\sqrt{{\ell}_{0_n}^2+H_R^2}$ represents
the distance from the $n$-th element of the IRS$_0$ to the typical user. Note that $\vert g_{{0,0_n}}\vert$ and $\phi_{0,0_n}$ represent the magnitude and phase component of the fading  channel from the $n$-th element of IRS$_0$ to the typical receiver. Similarly, ${ {f}}_{0_n,j}= \vert{ f}_{{0_n,j}}\vert e^{-j\psi_{0_n,j}}$ is the fading channel gain from the  $n$-th element of IRS$_0$ to $j$-th BS, thus ${\hat f}_{0_n, j}= \beta \left({t_{0_n, j}}\right)^{-\alpha/2} {f}_{0_n, j}$ and ${{\bf \hat  f}}_{0,j} \in \mathbb{C}^{N\times 1}$, where   {$${t_{0_n,j}}=\sqrt{{r^2_{0,0_n}}+d_j^2-2r_{0,0_n} d_j \cos(\angle t_{0_n,j} )}$$ represents
the distance from the $n$-th element of IRS$_0$ to the typical user, where $\angle t_{0_n,j}$ denotes the  angle opposite to $t_{0_n,j}$. }
Note that $\vert f_{{0_n,j}}\vert$ and   $\psi_{0_n,j}$ represent
the magnitude and phase component of the fading channel from $j$-th BS to $n$-th element of  IRS$_0$.  Finally, $\Theta_{0}$ denotes the phase shift of the IRS$_0$ and $\Theta_{0}={\rm{diag}}\{e^{j\theta_{0_1}},e^{j\theta_{0_2}},\cdots,e^{j\theta_{0_N}}\}$, and $\;C_{0_n,j}=\left({r_{0,0_n}}{t_{0_n,j}}\right)^{-\alpha/2}$.  

\subsubsection{Interference Power} The interference at a typical user in the IRS-assisted mode is composed of two parts (i) interference from the BSs, and (ii) interference from the IRSs. The aggregate interference from all the  BSs (excluding the nearest BS) is given as follows:
 		\begin{equation}
	    I_{B}=\sum_{j\in \Phi_{B}\backslash 0 }P \beta^2\vert h_{j}\vert ^2 d^{-\alpha}_{j}= \sum_{j\in \Phi_{B}  \backslash 0} P \beta^2\vert h_{j}\vert ^2 (\ell_j^2+H_B^2)^{-\alpha/2},
	\end{equation}
		 On the other hand, the aggregate interference from the IRSs can be modeled as follows:
\begin{align}\label{eq:IRSINt2}
\begin{split}
 I_R =&\sum_{j\in \Phi_{B} } \sum_{m= 1}^{M \backslash 0} P\;\vert\mathbf{ \hat g}^{H}_{0,m}  {\Theta_{m}} \;\mathbf{\hat f}_{m,{j}} \vert^2
=\sum_{j\in \Phi_{B}} \sum_{m= 1}^{M \backslash 0} P \vert\sum_{n=1}^{N} C_{m_n,j} f_{m_n,j} g_{0,m_n} {e^{j\theta_{m_n}} }\vert^2,
\end{split}
\end{align}
where $ g_{0,m_n}=\vert{ g}_{0,{m_n}}\vert e^{-j\phi_{0,m_n}}$ is the fading channel gain from the typical user to the $n$-th element of IRS $m$, thus ${\hat g}_{0,m_n}=\beta \left({r_{0,m_n}}\right)^{-\alpha/2} {g}_{0,m_n}$ and 
${{\bf \hat  g}}_{0,m} \in \mathbb{C}^{1\times N}$, where  $r_{0,m_n}=\sqrt{{\ell}_{m_n}^2+H_R^2}$ represents
the distance from $n$-th element of $m$-th IRS to the typical user. 
Note that $\vert g_{{0,m_n}}\vert$ and $\phi_{0,m_n}$ represent the magnitude and phase component of the fading  channel from $n$-th element of   $m$-th IRS to the typical receiver. 
Similarly, ${ {f}}_{m_n,j}= \vert{ f}_{{m_n,j}}\vert e^{-j\psi_{m_n,j}}$ is the fading channel gain from the  $n$-th element of IRS $m$ to $j$-th BS, thus ${\hat f}_{m_n, j}= \beta \left({t_{m_n, j}}\right)^{-\alpha/2} {f}_{m_n, j}$ and ${{\bf \hat  f}}_{m,j} \in \mathbb{C}^{N\times 1}$, where  {${t_{m_n,j}}=\sqrt{{r^2_{0,m_n}}+d_j^2-2r_{0,m_n} d_j \cos(\angle t_{m_n,j} )}$ represents
the distance from $n$-th element of $m$-th IRS to the typical user where $\angle t_{m_n,j}$ denotes the  angle opposite to $t_{m_n,j}$. }
Note that $\vert f_{{m_n,j}}\vert$ and   $\psi_{m_n,j}$ represent
the magnitude and phase component of the fading channel from $j$-th BS to $n$-th element of $m$-th IRS. Finally, $\Theta_{m}$ denotes the phase shift of the IRS and $\Theta_{m}={\rm{diag}}\{e^{j\theta_{m_1}},e^{j\theta_{m_2}},\cdots,e^{j\theta_{m_N}}\}$, and $\;C_{m_n,j}=\left({r_{0,m_n}}{t_{m_n,j}}\right)^{-\alpha/2}$.     

	\subsection{Signal and Interference Models (Direct Mode)}
		\subsubsection{Desired Signal Power}
	The signal power from the desired BS  to the typical user is:
	\begin{equation}
	\label{SigPowerDirect}
S_{D_0}=\hat{P} \beta^2\vert h_{0}\vert ^2 d^{-\alpha}_{0} = \hat{P} \beta^2\vert h_{0}\vert ^2 ({\ell_0^2+H_B^2})^{-\alpha/2},
	\end{equation}
	where $\hat{P}$ is the transmission power of the  BSs in direct mode, 
	 $h_0$ and $d_0$ are the small scale fading channel    and the  distance between the typical user to the nearest BS, respectively. 
	\subsubsection{Interference Power} The interference at a typical user in the direct mode is composed of two parts (i) interference from the BSs, and (ii) interference from the IRSs. The aggregate interference from the BSs (excluding the desired BS) is given as follows:
	\begin{equation}
	\label{eq:D_forDirect}
	  \hat I_{B}=\sum_{j\in \Phi_{B} \backslash 0} \hat{P} \beta^2\vert h_{j}\vert ^2 d^{-\alpha}_{j}= \sum_{j\in \Phi_{B} \backslash 0} \hat{P} \beta^2\vert h_{j}\vert ^2 (\ell_j^2+H_B^2)^{-\alpha/2},
	\end{equation}
	where $h_j$ and $d_j$ are the small scale fading channel    and the  distance between the typical user to the nearest BS, respectively.  On the other hand, the aggregate interference from all IRSs can be modeled as follows:
\begin{align}\label{eq:IRSINt}
\begin{split}
\hat I_{R}=&\sum_{j\in \Phi_{B}} \sum_{m= 1}^{M} \hat{P}\;\vert\mathbf{ \hat g}^{H}_{0,m}  {\Theta_{m}} \;\mathbf{\hat f}_{m,{j}} \vert^2
=\sum_{j\in \Phi_{B}} \sum_{m= 1}^{M} \hat{P} \vert\sum_{n=1}^{N} C_{m_n,j} {f_{m_n,j} g_{0,m_n}} {e^{j\theta_{m_n}} }\vert^2.
\end{split}
\end{align}


	 \subsection{Power Consumption  Model}
	\label{subsec:PowerModel}
	
We consider $p_{BS}$ and $p_U$ as the static power consumption of BS and user, respectively. The transmission power of BS in the direct mode is $\hat{P}$ and indirect mode is ${P}$. The IRS is acting as a passive device and does not have any additional transmission power consumption. However, the IRS power consumption is associated  with the number of IRS elements and the phase resolution \cite{huang2019reconfigurable} and can be written as  $p_{\rm IRS} ={N P_r(b)}$, where $P_r(b)$ denotes the phase resolution power consumption. The power consumption of the finite phase resolution, for instance, for 6 bits is $P_r(6)=$ 78 mW which is much lower than the power consumption for infinite phase resolution $P_r(\infty)=$ 45 dBm (Fig.~4 of \cite{huang2018energy}).  Therefore, hardware power consumption increases with an increase in resolution and the number of IRS elements    as provided in  \cite{huang2018energy,huang2019reconfigurable}.  
The system power consumption per user in the IRS-assisted mode  $p_{\rm ID}$ is given as 
$p_{\rm ID}= p_{\rm BS}+p_{\rm U}+ P+p_{\rm  IRS},$
whereas the power consumption of direct mode $p_{\rm D}$ is given as 
$p_{\rm D}= p_{\rm BS}+p_{\rm U}+ \hat{P}.$
 \subsection{Methodology of Analysis}
To derive the coverage probability of different types of users in a large-scale IRS-assisted network, our methodology is as follows: 
\begin{itemize}
\item  ({\bf IRS-assisted User}) Model the received signal power $S_{\rm {R_0}}$ as a sum of scaled generalized gamma random variables and then  derive the LT of $S_{R_0}$ (\textbf{Section~III}).
\item ({\bf IRS-assisted User}) 
Derive the LT of the aggregate interference observed at a typical IRS-assisted user from all BSs, i.e., LT of $I_{\rm B}$. Then, we model the  aggregate interference observed at a typical IRS-assisted user from all IRSs as sum of normal random variables and derive its corresponding LT, i.e., LT of $I_{\rm R}$ (\textbf{Section~IV}).
\item Then, apply Gil-Pelaez inversion to obtain $C_{\rm ID}$ conditioned on the distance  $r_{0,0}$\footnote{We approximate $r_{0,0_n}\approx r_{0,0}$ since the distance between the typical user and different elements of the nearest IRS is almost the same, i.e., the distance between IRS elements is negligible compared to the distance between the nearest IRS and the typical  user. Similarly $t_{0_n,j} \approx t_{0,j}$, $r_{0,m_n}\approx r_{0,m}$, and $t_{m_n,j} \approx t_{m,j}$.}.
       \item  ({\bf Direct User}) Derive the LT of $\hat{I}_{\rm B}$ and $\hat{I}_{\rm R}$, i.e., $\mathcal{L}_{\hat{I}_{\rm B}} (s)$ and $\mathcal{L}_{\hat{I}_{\rm R}} (s)$, respectively, and obtain    $C_{\rm D}$ conditioned on distance $d_0$.
    \item Derive the ergodic capacity  using Hamdi's lemma \cite{hamdi2010useful} and energy-efficiency of typical IRS-assisted user and direct user.
\end{itemize}  

 \section{ Statistics of the Received Signal Power  (IRS-assisted Transmission)}
 
In what follows, we   model the received  power at a typical IRS-assisted user  $S_{\rm {R_0}}$ as a a sum of scaled generalized gamma random variables and  derive the  LT of $S_{\rm {R_0}}$ conditioned on  $r_{0,0}$.
\begin{lem}
\label{propID}
The desired signal power through nearest IRS $S_{R_0}(r_{0,0})$ in \eqref{eq:desiredInDirect1} can be modeled as a sum of scaled generalized gamma random variable as follows: 
\begin{align}
\label{eq:DesiredSingProp}
\begin{split}
    S_{R_0}=P {r^{-{\alpha}}_{0,0}}{t^{-{\alpha}}_{0,j}} \sum_{q=1}^{N^2}  \vert a_q \vert X_{GG_{q}} (\zeta^2, 0.5 \kappa, 0.5),
\end{split}
    \end{align}
    where  $a_q=\cos(\beta_{0_n}-\beta_{0_k}), \forall q=1,\cdots,n+k,\cdots,N^2$, $n=\{1,\cdots,N\}$, $k=\{1,\cdots,N\}$.
\end{lem}
\begin{proof}
The desired signal power through nearest IRS $S_{R_0}(r_{0,0})$ in \eqref{eq:desiredInDirect1} is simplified using the following steps:
  \begin{align}
\label{eq:DesiredSingProof}
\begin{split}
    S_{R_0}(r_{0,0})=&P {r^{-{\alpha}}_{0,0}}{t^{-{\alpha}}_{0,j}}\vert\sum_{n=1}^{N}\vert  f_{0_n,j}\vert \vert g_{0,0_n}\vert{e^{-j\beta_{0_n}} }\vert^2\stackrel{(a)}{\approx}P{r^{-{\alpha}}_{0,0}}{t^{-{\alpha}}_{0,j}}\vert\sum_{n=1}^{N}  X_{G_n}(\kappa,\zeta) {e^{-j\beta_{0_n}} }\vert^2\\\stackrel{(b)}{=}&P {r^{-{\alpha}}_{0,0}}{t^{-{\alpha}}_{0,j}}\sum_{n=1}^{N}\sum_{k=1}^{N}  {\cos({\beta_{0_n}-\beta_{0_k}}) }X_{G_n}(\kappa,\zeta)X_{G_k}(\kappa,\zeta)\\\stackrel{(c)}{=}&P {r^{-{\alpha}}_{0,0}}{t^{-{\alpha}}_{0,j}} \sum_{n=1}^{N}\sum_{k=1}^{N}  {\cos({\beta_{0_n}-\beta_{0_k}}) } X_{GG_{k,n}} (\zeta^2, 0.5 \kappa, 0.5) \\\stackrel{(d)}{=}&P {r^{-{\alpha}}_{0,0}}{t^{-{\alpha}}_{0,j}} \sum_{q=1}^{N^2} { \vert a_q \vert} X_{GG_{q}} (\zeta^2, 0.5 \kappa, 0.5),
\end{split}
    \end{align}
   where $\;C_{0_n,j}\approx\;C_{0,j}$ (a) is followed by noting that $\vert{ g}_{0,{0_n}}\vert\vert{ f}_{{0_n,j}}\vert$ is the product of two independent Rayleigh distributed random variables with mean and variance $\mu_x=\sigma \pi/2$ and $\sigma^2_x=2^2\sigma^2(1-\pi^2/16)$, respectively. However, the exact distribution of the product of two i.i.d Rayleigh random variables in \cite{salo2006distribution} is complicated. Therefore, to maintain tractability, we approximate it as  a gamma random variable $X_{G_n}(\kappa,\zeta)$ with the shape and scale parameter $\kappa= m= 1.6467$ and $\zeta=\frac{\Omega}{m} = 0.9539 $, respectively \cite{lu2011accurate}. 
    Note that (b) follows from the simplification of (a)  using $\vert x\vert^2=\rm Re(x)^2+\rm Im(x)^2$ and  trigonometric identity $\cos(\alpha-\beta)=\cos(\alpha)\cos(\beta)-\sin(\alpha)\sin(\beta)$. Next, (c) follows from the fact that the product of two i.i.d gamma random variables is equivalent to generalized gamma random variable  $X_{GG_q}(a,d,p)$, where $a=\zeta^2,d=\frac{\kappa}{2},$ and $p=\frac{1}{2}$ represent the scale, shape, and power, respectively \cite{stacy1962generalization}. Finally, in (d), the double summation $n=1,\cdots,N$, $k=1,\cdots,N$ and $a_q=\cos(\beta_{0_n}-\beta_{0_k})$ is transformed to single summation  $q=1,\cdots,n+k,\cdots,N^2$, where  $-\pi/2 \leq \beta_{0_n}-\beta_{0_k} \leq \pi/2$.
\end{proof}
In what follows, we derive the conditional LT of the received signal power in the IRS-assisted communication mode.
	\begin{lem}
	\label{lemma_desired_ID}
	
   Conditioned on $r_{0,0}$, the LT of the   $S_{R_0}$ experienced by the typical user through nearest IRS  $\mathcal{L}_{S_{R_0}}$ in the IRS-assisted indirect communication mode  is given as follows:
	  \begin{align}
	  \label{lapsro}
        \begin{split}
                \mathcal{L}_{ S_{R_0}|r_{0,0}}(s)=&\mathbb{E}\left[ \prod_{q=1}^{N^2}             \frac{1}{\left( 2\zeta^2s \;\hat{a}_q  \right)^d } \exp\left(\frac{1}{8 \;\zeta^2 s \;\hat{a}_q   }\right) D_{\rm-2d}\left(\frac{1}{2\zeta^2s \;\hat{a}_q   }\right)\right],
                   \end{split}
    \end{align}
    where $\hat{a}_q=P {r^{-{\alpha}}_{0,0}}{t^{-{\alpha}}_{0,j}}  a_q $, $a_q$ is defined in { Lemma~\ref{propID}} and $D_{-v}(\cdot)$ is the parabolic cylinder function.
	\end{lem}
	\begin{proof}
	  The LT of $S_{R_0}$ is given by using step (d) of \eqref{eq:DesiredSingProof} as follows:
    \begin{align}
    \label{eq:nnn}
        \begin{split}
                \mathcal{L}_{S_{R_0}|r_{0,0}}(s)&=\mathbb{E}[e^{-s S_{R_0}}]
                \\&= \mathbb{E}[e^{-sP {r^{-{\alpha}}_{0,0}}{t^{-{\alpha}}_{0,j}} \sum_{q=1}^{N^2}  a_q X_{GG_q}(\zeta^2, 0.5 \kappa, 0.5) }]
                \\&=\prod_{q=1}^{N^2} \mathbb{E}[e^{-sP {r^{-{\alpha}}_{0,0}}{t^{-{\alpha}}_{0,j}}  a_q X_{GG_q}(\zeta^2, 0.5 \kappa, 0.5) }],
                   \end{split}
    \end{align}
    where \eqref{eq:nnn} follows from the fact that MGF of the linear combination of independent variables can be rewritten as the product of MGFs of each of the independent variables, and
      \begin{align}
      \label{eq:mmm}
        \begin{split}
               & \mathbb{E}[e^{-sP {r^{-{\alpha}}_{0,0}}{t^{-{\alpha}}_{0,j}}  a_q X_{GG_q}(\zeta^2, 0.5 \kappa, 0.5) }]=\int_{0}^{\infty}e^ {-s\hat{a}_q  X_{GG_q}} f_{X_{GG}}(x)dx\\&\stackrel{(a)}{=} \int_{0}^{\infty} \frac{0.5}{\zeta^{ \kappa} \Gamma( \kappa)}x^{0.5 \kappa-1} e^{-s\hat{a}_q    x-\sqrt{\frac{x}{\zeta^2}}}dx
               \\&\stackrel{(b)}{=}\frac{2}{\zeta^{ \kappa} \Gamma( \kappa)} \int_{0}^{\infty} g^{ \kappa-1} e^{-s\hat{a}_q   g^2-{\frac{g}{{\zeta}}}}dg
               \\&\stackrel{(c)}{=}
               \frac{1}{\left( 2\zeta^2s \;\hat{a}_q  \right)^{0.5 \kappa} } \exp\left(\frac{1}{8 \;\zeta^2 s \;\hat{a}_q   }\right) D_{\rm- \kappa}\left(\frac{1}{2\zeta^2s \;\hat{a}_q   }\right),
        \end{split}
    \end{align}
    where (a) is obtained by substituting the probability density function of GG random variable 
    $f_{XX_G}(x)= \frac{0.5}{\zeta^{\kappa} \Gamma(\kappa)}x^{0.5\kappa-1} e^{-\sqrt{\frac{x}{\zeta^2}}}$ \cite{lu2011accurate}, (b) is   obtained by changing variable $g=\sqrt{x}$, (c) is derived by using the identity $\int_{0}^{\infty} g^{\nu-1} e^{-\beta  g^2-\gamma g} dg=(2\beta)^{-\frac{\nu}{2}} \Gamma [\nu] \exp(\frac{\eta^2}{8 \beta}) D_{-\nu}\left(\frac{\eta}{\sqrt{2 \beta}}\right)$ from  Eq. 3.462 of \cite{jeffrey2007table}, where
    $D_{-\nu}\left(.\right)$ represents the  parabolic cylinder function. Finally, by using   $\hat{a}_q=P {r^{-{\alpha}}_{0,0}}{t^{-{\alpha}}_{0,j}}  a_q $ in  step (c) of \eqref{eq:mmm} and \eqref{eq:nnn} results in {\bf Lemma~}\ref{lemma_desired_ID}.
    	\end{proof}
 	\begin{figure}[t]
	\begin{center}
	    \includegraphics[scale=0.74]{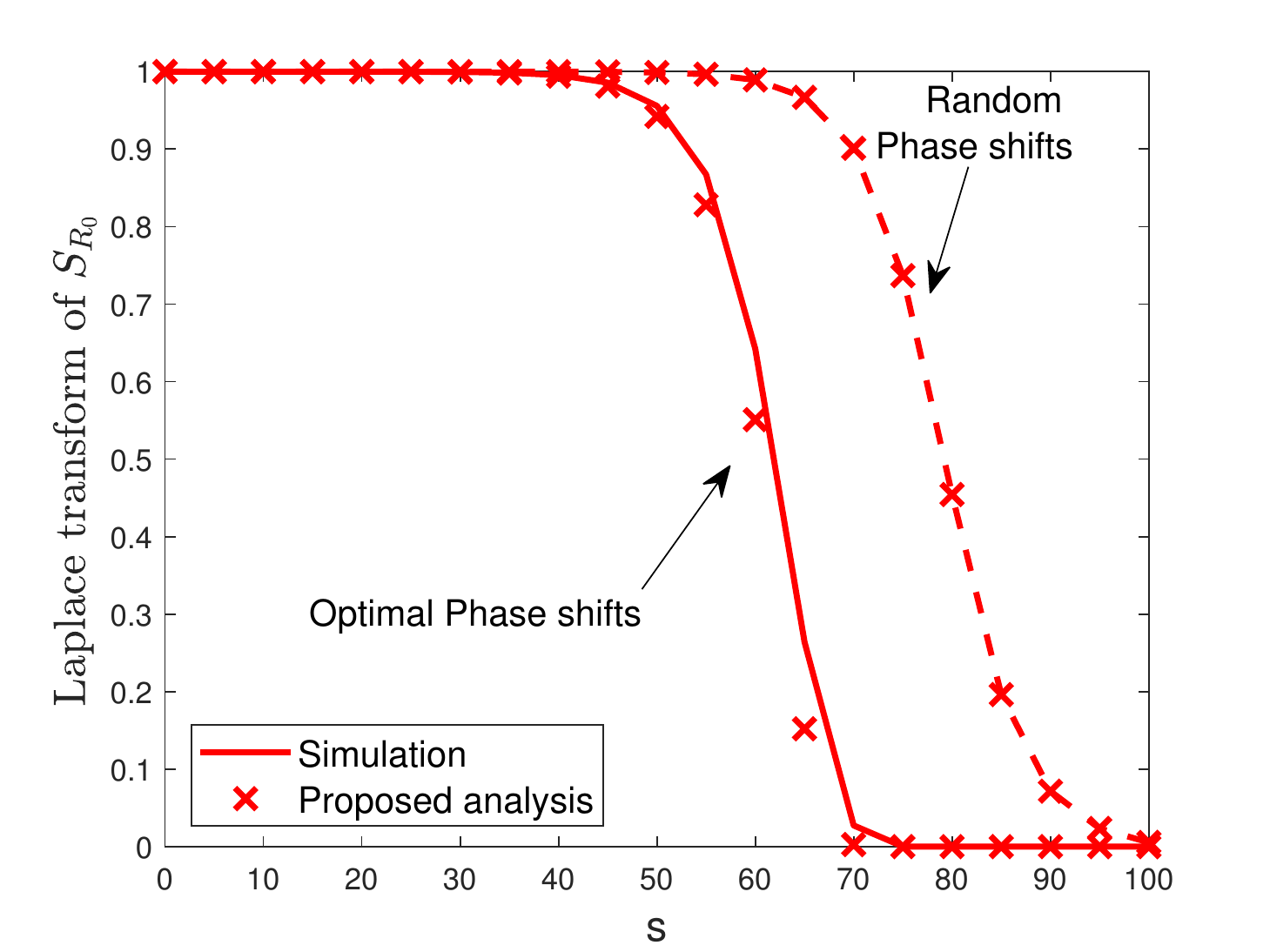}
			\caption{Validation of conditional LT in \eqref{eq:DesiredSingProp} of the desired received signal power  $S_{R_0}(r_0)$ considering  (i) random IRS phase shifts and (ii) optimal IRS phase shifts obtained from CVX,   using Monte-Carlo simulations.}
	\end{center}			\label{fig:Figure_MGF_SR0}	
	\end{figure}
Fig.~2 validates the accuracy of the LT of the received signal power (derived in {\bf Lemma~2}) of the typical  IRS-assisted user with  the
Monte-Carlo simulations. Our derived expressions match well with the simulations confirming the accuracy of our $S_{R_0}$ model and its corresponding LT. In both Lemma~2 and simulations, the phase-shifts are obtained optimally from CVX. Specifically, we solve the following problem ({\bf P1}) to maximize the received signal power (given in \eqref{eq:desiredInDirect1}) and obtain the optimal phase-shifts:
	\begin{align}
	\begin{split}
	{\bf P1}:	&\max_{\theta_{0_n}, \forall {n}}\; 	 {\rm S_{R_0}}= P \vert\sum_{n=1}^{N} C_{0_n,j} f_{0_n,j}  g_{0,0_n} {e^{j\theta_{0_n}} }\vert^2 \\
	{\rm s.t.}\; & 0\le \theta_{0_n}\le \pi, \forall {n=1,\cdots,N.}, \;\;\;\; 
	\end{split}
	\label{eq:OPtProb1Var1in_d}
	\end{align}

	Substituting $f_{0_n,j}=\vert{ f}_{{0_n,j}}\vert e^{-j\psi_{0_n,j}}$, $  g_{0,0_n}=\vert{ g}_{0,{0_n}}\vert e^{-j\phi_{0,0_n}}$ , $\Theta_{0}={\rm{diag}}\{e^{j\theta_{0_1}},e^{j\theta_{0_2}},\cdots,e^{j\theta_{0_N}}\}$  and $\;C_{0_n,j}\approx\;C_{0,j}={r^{-{\alpha}}_{0,0}}{t^{-{\alpha}}_{0,j}}$ defined in Sec.~IIB, the objective function  can be rewritten as $P {r^{-{\alpha}}_{0,0}}{t^{-{\alpha}}_{0,j}}\vert\sum_{n=1}^{N}\vert  f_{0_n,j}\vert \vert g_{0,0_n}\vert{e^{-j\beta_{0_n}} }\vert^2$. Since  $P {r^{-{\alpha}}_{0,0}}{t^{-{\alpha}}_{0,j}}$ is independent of the optimization variable, we can discard this term. Now, we transform the objective to equivalent matrix form as $
\vert\mathbf{ \tilde g}^{H}_{0,0} \; {\bf B_{0}} \;\mathbf{\tilde f}_{0,{j}} \vert^2$, where  $\mathbf{ \tilde g}_{0,0} \in \mathbb{R}^{1\times N},$ $\mathbf{ \tilde f}_{0,j} \in \mathbb{R}^{N\times 1},$  and ${\bf B}_0={\rm{diag}}\{e^{j\beta_{0_1}},e^{j\beta_{0_2}},\cdots,e^{j\beta_{0_N}}\}$. Since the  objective function is a scalar,  we can convert absolute square to norm square as $
\Vert\mathbf{ \tilde g}^{H}_{0,0}\;  {\bf B_{0}} \;\mathbf{\tilde f}_{0,{j}} \Vert^2$.  Finally, defining $\mathbf{v}=[v_1,\cdots,v_n]^H$, where $v_n=e^{j\beta_{0_n}}, \forall n$, and  $\Phi=\rm{diag}({\bf \tilde g}^{H}_{0,0}) \mathbf{\tilde f}_{0,{j}}$, we reformulate
$\Vert\mathbf{ \tilde g}^{H}_{0,0}  {B_{0}} \;\mathbf{\tilde f}_{0,{j}} \Vert^2= \Vert \mathbf{v}^H \Phi\Vert^2$. The problem \textbf{P1} can thus  be reformulated as follows:
			\begin{align}
	\begin{split}
	{\bf P2}:	&\max_{\mathbf{v}}\; 	\mathbf{v}^H \Phi \Phi^H \mathbf{v} \\
	{\rm s.t.}\; & \vert v_n\vert^2=1, \forall {n=1,\cdots,N}. \;\;\;\; 
	\end{split}
	\label{eq:OPt2}
	\end{align}  
	\textbf{P2} is non-convex  quadratically constrained quadratic program (QCQP) in the homogeneous form and the constraint is rank one \cite{boyd2004convex}. Now, defining $\mathbf{V}=\mathbf{v}\mathbf{v}^H $, we apply semi-definite relaxation (SDR) to relax the constraint as follows:
	\begin{align}
	\begin{split}
	{\bf P3}:	&\max_{\mathbf{V}}\; 	\mathrm{Tr}\;(\Phi \Phi^H\mathbf{V}  )\\
	{\rm s.t.}\; & \mathbf{V}_{n,n}=1, \forall {n=1,\cdots,N}, \;\;\;\;  \mathbf{V}\ge 0.
	\end{split}
	\label{eq:OPt3}
	\end{align} 
	Since the problem is now transformed in to a convex semidefinite program (SDP), similar to \cite{wu2019intelligent}, we  solve it for the optimal value using CVX.

Furthermore, Fig.~2 also compares the LT of $S_{R_0}$ with optimal IRS phase-shifts to  LT of $S_{R_0}$ with random IRS phase shifts. For a given value of $s$, the LT  of $S_{R_0}$ with optimal IRS phase-shifts  is lower than the LT  of $S_{R_0}$ with random IRS phase-shifts.
Thus, it is evident that the received signal power $S_{R_0}$ with optimal phase-shifts significantly outperforms the received signal power $S_{R_0}$ with  random phase shifts. 

As a special case of {\bf Lemma~2} for optimal phase-shifts, the  statistics of the received signal power can be modeled as follows. 
\begin{corollary}	
\label{cor:meanSR0}
{The optimal received signal power  can be obtained if we substitute  $\beta_{0_n,j}=\theta_{0_n}-\psi_{0_n,j}-\phi_{0,0_n}=0 $ in \eqref{eq:DesiredSingProp}, which maximizes   $a_q$ to  unity $\forall n\in\{1,\cdots,N\}$  \cite{basar2019wireless} and results in maximum $S_{R_0}$ as $S_{R_0}=P {r^{-{\alpha}}_{0,0}}{t^{-{\alpha}}_{0,j}} W$. In this case, $W= \sum_{q=1}^{N^2} X_{GG_{q}} (\zeta^2, 0.5 \kappa, 0.5) $ can be modeled as a normal random variable  Since the square of the number of IRS elements can be a large number,  using CLT with mean $
\mu_w= N^2 \mu_{GG}$ and variance $
\sigma^2_w= N^2 \sigma^2_{GG}$. 
where $\mu_{GG}=
\zeta^2{\frac  {\Gamma ( \kappa+2)}{\Gamma (\kappa)}}$
and 
$ \sigma_{GG}= \zeta^{4}\left({\frac  {\Gamma ( \kappa+4)}{\Gamma (\kappa)}-\mu^2_{GG} }\right)$ \cite{stacy1962generalization}, we have the mean and the variance of $S_{R_0}$ as  $\mathbb{E}[S_{R_0}]=P {r^{-{\alpha}}_{0,0}}{t^{-{\alpha}}_{0,j}}\mu_{GG}$ and  $\mathbb{V}[S_{R_0}]=P^2 {r^{-{2\alpha}}_{0,0}}{t^{-{2\alpha}}_{0,j}} \sigma_{GG}$}, respectively.
\end{corollary}	
    	



\section{Statistics  of the Aggregate Interference (IRS-assisted Transmission)}
In this section, we first derive the LT of the aggregate interference observed at a typical IRS-assisted user from all BSs. Then, we model the worst-case aggregate interference observed at a typical IRS-assisted user from all IRSs and derive its corresponding LT.

The LT of the aggregate interference observed at a typical IRS-assisted user from all BSs (excluding the blocked nearest direct BS) $\mathcal{L}_{I_B}(s)$ is derived as follows: 

\begin{align}
\label{eq:IBDirec}
\begin{split}
\mathcal{L}_{ I_B|d_{0}}(s) =&\mathbb{E}[ e^{-s\sum_{j\in \Phi_{B} \backslash 0} P \beta^2\vert h_{j}\vert ^2 (\ell_j^2+H_B^2)^{-\alpha/2}}]
\\ \stackrel{(a)}{=} & \: \mathbb{E}_{\Phi_{\rm B}}\left[\prod_{j\in \Phi_{B} \backslash 0} \frac{1}{1+s \; P \beta^2 (\ell_j^2+H_B^2)^{-\alpha/2} }\right]
    \\\stackrel{(b)}{=} & \: \mathrm{exp}\left({-2 \pi \lambda_B \int_{\ell_0}^{\infty} \left( 1-
\frac{1}{1+ P \beta^2 (\ell_j^2+H_B^2)^{-\alpha/2}}\right)\ell_j d\ell_j}\right)
  \\\stackrel{(c)}{=} &\: \mathrm{exp}\left(-2 \pi \lambda_B \int_{d_{0}}^{\infty} \left( 1-
\frac{1}{1+  P \beta^2  d^{-\alpha}_{j} }\right) d_j \;dd_j \right),
\end{split}
\end{align}
where (a) is obtained by applying the LT of $\vert h_{j}\vert^2$ and $\vert h_{j}\vert^2\sim\exp(1)$, and (b) is derived using PGFL w.r.t the two-dimensional distance $\ell_j$ of the interfering BSs \cite{lyu2020hybrid}, and (c) is obtained by substituting  $d_j=\sqrt{\ell_j^2+H_B^2}$. The closed-form expression can then be  obtained as follows:
 \begin{align}
\label{eq:LTDirectforID}
\begin{split}
\mathcal{L}_{ I_B|d_0}(s )= \mathrm{exp}\left(- 2 \pi \lambda_B \frac{ d_{0}^{2-\alpha} s P \beta^2  }{\alpha-2} \;_2 F_1\left(1,\frac{-2+\alpha}{\alpha};2-\frac{2}{\alpha}; -s\; P \beta^2 d_{0}^{-\alpha} \right) \right).
\end{split}
\end{align}

\begin{corollary} For $\alpha=4$. the LT of the aggregate interference to the typical user through all the BSs (except the associated BS$_0$ ) $ \mathcal{L}_{\hat{I}_B}$ in the the  direct  mode can simplified as:
$$\mathcal{L}_{\hat I_B|d_0}(s )= \mathrm{exp}\left( - \pi \lambda_B\sqrt{ s P \beta^2  } \;{ {\arctan}\left(\sqrt{ s\; \hat{P} \beta^2 d_{0}^{-4}} \right)}\right),
$$ using $ \;_2 F_1\left(1,0.5;1.5; -X^2 \right) =\frac{{ \arctan}X}{{X}}$, for $ \vert X\vert<1$ \cite{NIST:DLMF}[Eq.~15.4.3].

    	\end{corollary}

\begin{lem}[Lower Bound on the Aggregate Interference from Multiple IRSs] We  reformulate the aggregate interference observed at a typical user from all IRSs (excluding nearest IRS) in a multi-IRS, multi-BS scenario as 
$
{ I}_{\rm R}{\leq} \sum_{j\in \Phi_{B }} P Z_j,
$
where
$ Z_j=\sum_{m=1}^{M-1} r_{{0,m}}^{-\alpha}\; t_{{m,j}}^{-\alpha}\;Y_m$.
\end{lem}
\begin{proof}
Taking $\beta_{m_n,j}=\theta_{m_n}-\psi_{m_n,j}-\phi_{0,m_n}$, the ${ I}_{\rm R}$ expression in \eqref{eq:IRSINt2} can be rewritten as follows:
\begin{align}
\begin{split}\label{eq:LapID}
{ I}_{\rm R}
\stackrel{(a)}{=}&\sum_{j\in \Phi_{B }}  P\sum_{m=1}^{M-1}   {r^{-{\alpha}}_{0,m}}{t^{-{\alpha}}_{m,j}} \quad\vert\sum_{n=1}^{N} \vert f_{m_n,j} \vert \; \vert g_{0,m_n} \vert{e^{j\beta_{m_n,j}} }\vert^2\\
\stackrel{(b)}{\leq}&\sum_{j\in \Phi_{B}}  P\sum_{m=1}^{M-1}    {r^{-{\alpha}}_{0,m}}{t^{-\alpha}_{m,j}} \vert\sum_{n=1}^{N} \vert f_{m_n,j} \vert \; \vert g_{0,m_n} \vert \vert^2\\
\stackrel{(c)}{=}&\sum_{j\in \Phi_{B}} P \sum_{m=1}^{M-1}   {r^{-{\alpha}}_{0,m}}{t^{-\alpha}_{m,j}} Y_m
\stackrel{(d)}{=}\sum_{j\in \Phi_{B}} P Z_j,
\end{split}
\end{align}
where  (a) is obtained by substituting $\;C_{m_n,j}=\left({r_{0,m_n}}{t_{m_n,j}}\right)^{-\alpha/2}$ and considering the approximation ${r_{0,m}}\approx{r_{0,m_n}}$, ${t_{m,j}}\approx{t_{m_n,j}}$ as discussed in footnote-1,  (b) follows from  $\beta_{m_n,j}=\theta_{m_n}-\psi_{m_n,j}-\phi_{0,m_n}=0$ which results in the maximum interference  (excluding nearest IRS) and hence referred to as \textit{worst case interference}. Finally,
step (c) and step (d) follow by defining  $Y_m=\vert\sum_{n=1}^{N} \vert f_{m_n,j} \vert \; \vert g_{0,m_n} \vert \vert^2$ and $ Z_j=\sum_{m=1}^{ 
 M-1}  r_{{0,m}}^{-\alpha}\; t_{{m,j}}^{-\alpha}\;Y_m$, respectively.
\end{proof}
In what follows, we derive the statistics of the aggregate interference observed at a typical user from multiple IRSs in a multi-BS scenario.
\begin{lem}[Distribution of the Aggregate Interference from Multiple IRSs  (Excluding the Nearest IRS)  in a Multi-BS Scenario] Leveraging the results in {\bf Lemma~3}, given
$
{ I}_{\rm R}{\leq} \sum_{j\in \Phi_{B}} P Z_j,
$
where
$ Z_j=\sum_{m=1}^{M-1} r_{{0,m}}^{-\alpha}\; t_{{m,j}}^{-\alpha}\;Y_m$ follows a Normal distribution with  mean and variance given by
 $$\mu_{Z_j}= \mathbb{E}[ r_{{0,m}}^{-\alpha}]  ((M-1) {t^2_{j}})^{-\alpha/2}\; (1+\lambda) \quad \mathrm{and} \quad \sigma^2_{Z_j}=2  \mathbb{V} [ r_{{0,m}}^{-\alpha}]  ((M-1) {t^2_{j}})^{-\alpha}(1+2\lambda),$$
\label{Prop2eq}
and $ Y_m$ represents the non-central Chi-square random variable with mean and variance $\mu_Y=(1+\lambda)$ and $\sigma_Y^2=2(1+2\lambda)$, respectively.
\label{lemmaIR}
\end{lem}

 \begin{proof}
 Let $X_n= \vert{ g}_{0,{m_n}}\vert\vert{ f}_{{m_n,j}}\vert$ denote the product of two independent Rayleigh distributed random variables with mean and variance $\mu_x=\sigma \pi/2$ and $\sigma^2_x=2^2\sigma^2(1-\pi^2/16)$, respectively \cite{shafique2020optimization}. Since the IRS elements are typically large, we leverage on  central limit theorem (CLT) to depict 
 $X^\prime=\sum_{n=1}^N {X_n}$   follows a normal distribution with the mean and variance given by $\mu_{X^\prime}=N \mu_{X}$ and $\sigma^2_{X^\prime} =N \sigma^2_{X}$, respectively. We refer to this approximation  as  {\em Level-1 Gaussian approximation}. Consequently, $Y_m=\vert \sum_{n=1}^{N}  X_n  \vert^2$ will follow a non-central chi-square distribution  with unity degree of freedom $\nu=1$  and non-centrality parameter $\lambda=\frac{1}{2}\frac{\mu_{X^\prime}}{\sigma^2_{X^\prime} }$ \cite{shafique2020optimization}. Therefore, the mean and variance of $Y_m$ can be obtained as in\textbf{ Lemma~\ref{Prop2eq}}.

 Let $Y^\prime_m=  r_{{0,m}}^{-\alpha}\; t_{{m,j}}^{-\alpha}\;Y_m$ denote the product of three random variables ${t^{-\alpha}_{m,j}} $, $r_{{0,m}}^{-\alpha} $, and $Y_m$, where ${t^{-\alpha}_{m,j}} $, and  $r_{{0,m}}^{-\alpha} $ are correlated by cosine law as
${t^{-\alpha}_{m,j}}=\left({r^2_{0,m}}+d_j^2-2r_{0,m} d_j \cos\psi_{m}\right)^{-\alpha/2}$ \cite{kishk2020exploiting,zhu2020stochastic}. 
To simplify the analysis, we propose  an alternate formulation of  $t_{\rm {m,j}}$, i.e., instead of using  {\rm cosine law} we alternatively define $ {t_{m,j}}=\sqrt{\ell_{m,j}^2+(H_B-H_R)^2}$ (refer to the triangle in  Fig.~\ref{fig:Tania1}(b)).
Next, to enhance tractability, we consider that the typical IRS is located in the middle of the $BS_{\rm j}$ and typical user (i.e., $\ell_{\rm m,j}\approx \frac{\ell_{\rm j}}{2}$ ) which upon substitution gives 
\begin{equation}\label{tj}
   t_{\rm m,j}\approx t_{\rm j}=\sqrt{\left(\frac{\ell_{\rm j}}{2}\right)^2+(H_B-H_R)^2}.
\end{equation}  
 Subsequently, we have $Y^\prime_m\approx  r_{{0,m}}^{-\alpha}\; t_{{j}}^{-\alpha}\;Y_m$, and  $Z=\sum_{m=1}^{M-1} Y^\prime_m$ will follow a normal distribution using CLT as shown in \textbf{Lemma~4}. We refer to this as {\em Level-2 Gaussian approximation}. 
\end{proof}
	\begin{figure}[!t]
		\begin{center}			\includegraphics[scale=0.5]{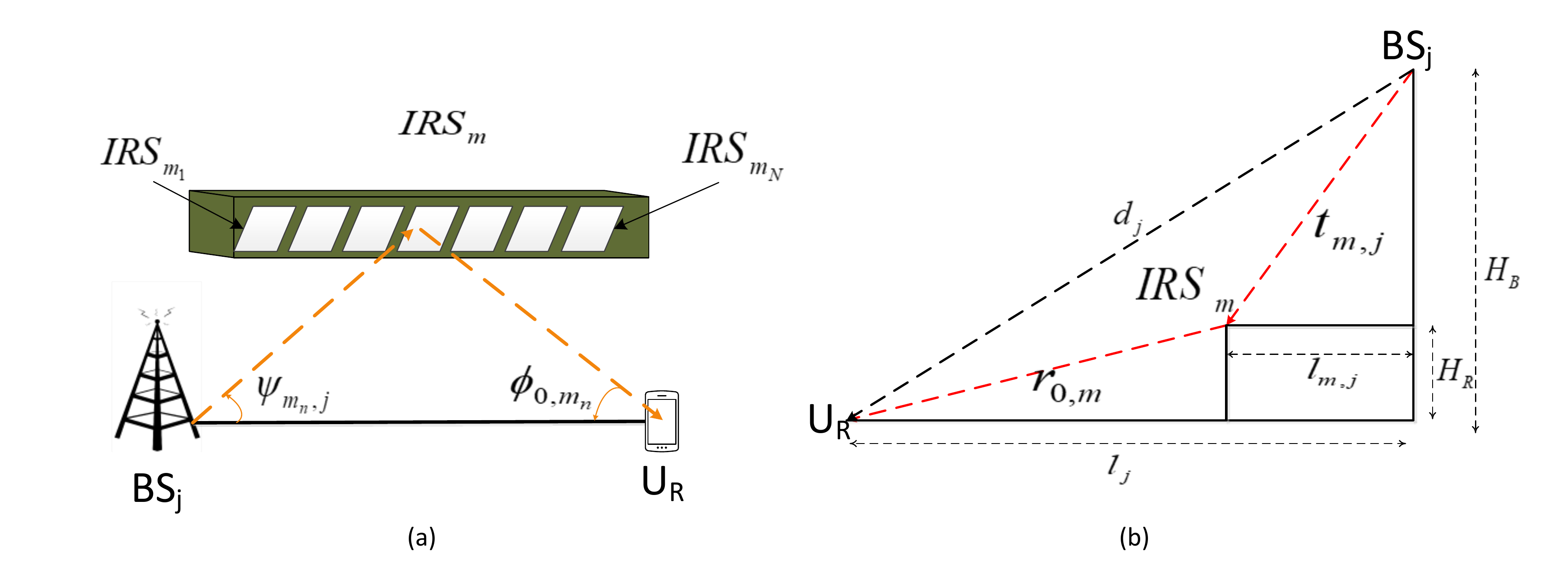}
			\caption{(a): Zoomed view of IRS functionality as a reflector, and (b) triangle explaining the  IRS distance approximation. }			\label{fig:Tania1}
		\end{center}
	\end{figure}

{ 
The factor $r_{{0,m}}^{-\alpha} t_{{m,j}}^{-\alpha} $  is  important  in modeling $Y^\prime_m$  as is evident in \textbf{Lemma~4}. Note that $ r_{{0,m}}^{-\alpha}$ and $t_{{m,j}}^{-\alpha}$ are correlated using cosine law. However, Fig.~4 shows that the correlation is weak  and thus  the approximation in \eqref{tj} is accurate. In the sequel, we first compare $\mathbb{E}[r_{{0,m}}^{-\alpha}\;] \mathbb{E}[ t_{{m,j}}^{-\alpha}]$, and $\mathbb{E}[r_{{0,m}}^{-\alpha} t_{{j}}^{-\alpha}]$ to show the weak correlation. Then, we demonstrate the validity of the proposed approximation   $\mathbb{E}[  r_{{0,m}}^{-\alpha}\;]\mathbb{E}[  t_{{j}}^{-\alpha}]$ to validate its accuracy in
Fig.~\ref{fig:Figure1_MeanApprox_Comparison_Final2}. It is also clear from the figure that $\lambda_R$ does not have any impact on the distances $t_{{m,j}}$ and $r_{{0,m}}$ on average.  It is clear from the right figure that increase in path-loss exponent $\alpha$ causes an increase in the path-loss distance term and hence decreases in $\mathbb{E}[r_{{0,m}}^{-\alpha}] \mathbb{E}[ t_{{m,j}}^{-\alpha}]$, $\mathbb{E}[ r_{{0,m}}^{-\alpha} t_{{j}}^{-\alpha}]$ and   $\mathbb{E}[  r_{{0,m}}^{-\alpha}\;]\mathbb{E}[  t_{{j}}^{-\alpha}]$ are evident.}
	\begin{figure}[ht!]
	\centering
	
	\includegraphics[scale=0.75]{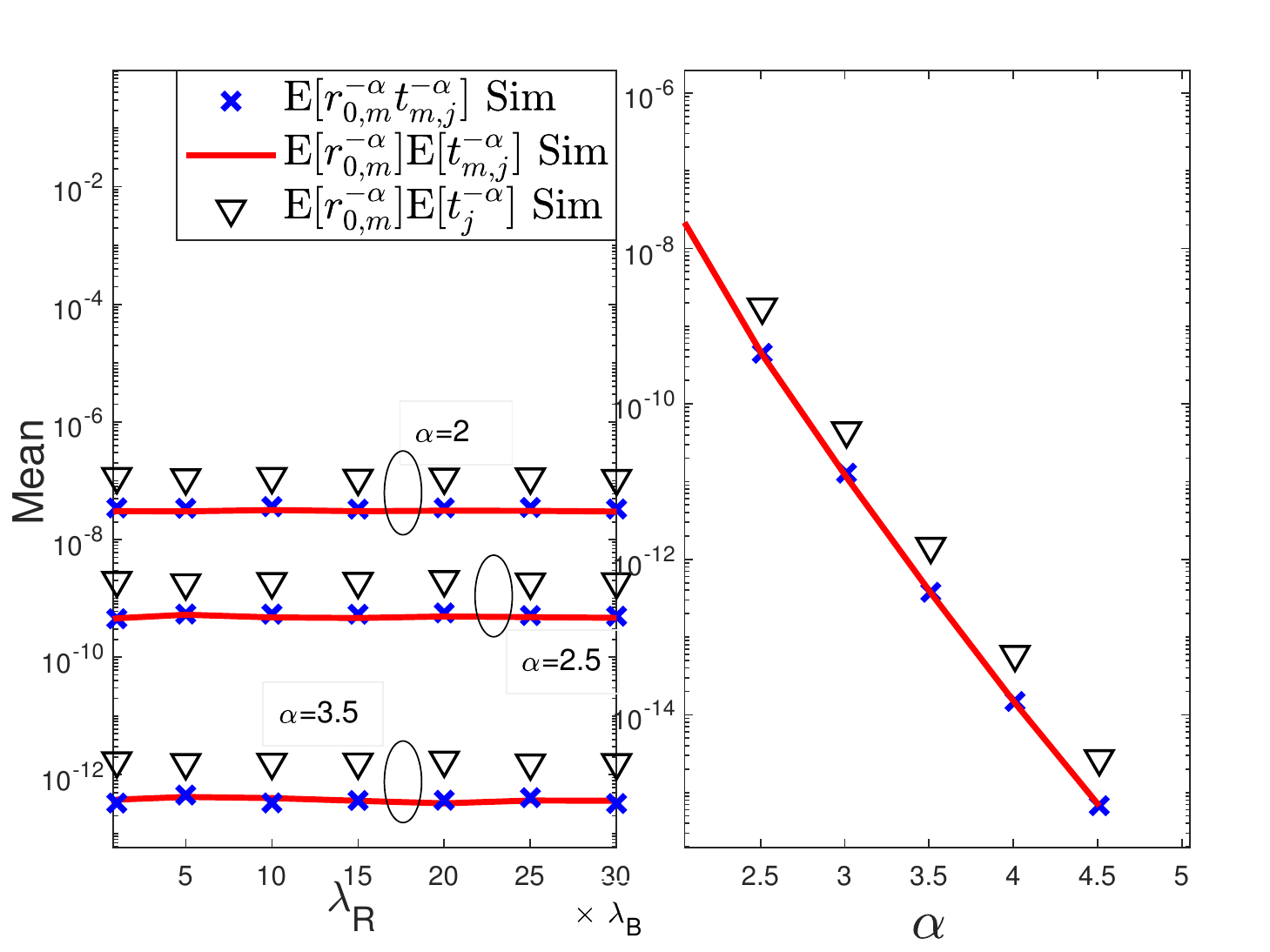}
	\caption{Comparison of   $\mathbb{E}[ r_{{0,m}}^{-\alpha} t_{{j}}^{-\alpha}] $,  $\mathbb{E}[r_{{0,m}}^{-\alpha}] \mathbb{E}[ t_{{m,j}}^{-\alpha} ] $ and the proposed approximation of $\mathbb{E}[  r_{{0,m}}^{-\alpha}\;]\mathbb{E}[  t_{{j}}^{-\alpha}]$ in  \eqref{tj}.}
	\label{fig:Figure1_MeanApprox_Comparison_Final2}	 
\end{figure}
In what follows, we derive the first  and second moment of $ r_{{0,m}}^{-\alpha} $ as is required in {\bf Lemma~4}.
\begin{lem}
\label{lemma_r_mean}
The $i$-th moment of the random variable $r_{{0,m}}^{-\alpha}$ can be derived for finite values of $R$ and when $R \rightarrow \infty$, respectively, as follows:

\begin{align}
\begin{split}
\label{rmena}
&\mathbb{E}[ (r_{{0,m}}^{-\alpha})^i]=\int_{0}^{R} r_{{0,m}}^{-i \alpha} \frac{\ell_{m}^2}{\pi R^2}  d\ell_{m} =\int_{\ell_{m}=0}^{R}({\ell_{m}^2+H_R^2})^{-\frac{i\alpha}{2}} \frac{\ell_{m}^2}{\pi R^2}  d\ell_{m}
\\&=\frac{-2 \left({H_R^2+ R^2}\right)^{-1-\frac{i\alpha}{2}} }  {(-2 +i \alpha) R^2}+
\frac{2 H_R^{2-i\alpha} }  {(-2 +i \alpha) R^2} 
\\&\mathrm{lim}_{R \rightarrow \infty} \mathbb{E}[ (r_{{0,m}}^{-\alpha})^i]=\frac{2 (H_R)
^{2-i\alpha}}  {(-2 +i \alpha) R^2}.
\end{split}
\end{align}
\end{lem}
Finally, we derive the  LT of  the interference experienced by the typical user from all IRSs to compute the coverage probability.
\begin{lem}
  \label{lammaLTDIDforID}
    The LT of  interference experienced by the typical user through the  all IRSs (except nearest IRS to the typical user) $\mathcal{L}_{{I}_R}$ in the IRS- assisted communication mode  is given as follows:
    \begin{align}
\label{eq:LTInDirectDforID}
\begin{split}
\mathcal{L}_{ I_{\rm R}|r_{0,0}}(s )\approx \exp \left( 2 \pi \lambda_B\frac{4}{\alpha} \sum_{i=1}^{\infty} \frac{ b_i(s)}{{i-\frac{2}{\alpha}}} \;\left( {X_R^{i-\frac{2}{\alpha}}}-{X_0^{i-\frac{2}{\alpha}}} \right) \right),
\end{split}
\end{align}  
where $b_i(s)$ denotes the Taylor's series expansion coefficients of $\exp(-k_1(s) x-k_2(s) x^2)$ and $k_1(s)=\mu_{Z_j} s \; P/t_j^{-\alpha} $, and  $k_2(s)=\frac{1}{2 t_j^{-2\alpha}} \sigma^2_{Z_j} s^2  P^2$.

\end{lem}
\begin{proof}
Using \eqref{eq:LapID} in \eqref{eq:IRSINt2}, we derive LT $\mathcal{L}_{ {I}_R}(s)$ as: 
\begin{align}
\begin{split}\label{eq:LapIDProof}
\mathcal{L}_{ I_R|{r_{0,0}}}(s)&=
\mathbb{E}[ e^{-s  I_R}]=\mathbb{E}[e^{-s  \sum_{j\in \Phi_{B }}  P Z_j}]=\mathbb{E}[ \prod_{j\in \Phi_{B }} e^{-s \;P Z_j}]\\ \stackrel{(a)}{=} & \mathbb{E}[ \prod_{j\in \Phi_{B }}\mathbb{E}_Z[ e^{-s \; P Z_j}]]=\mathbb{E}[ \prod_{j\in \Phi_{B }}  e^{-(\mu_{Z_j} s \; P +\frac{1}{2} \sigma^2_{Z_j} s^2  P^2)  }]\\ \stackrel{(b)}{=} &\exp \left(- 2 \pi \lambda_B \int_{0}^{R}\left(1- e^{-(k_1(s) t^{-\alpha}  +k_2(s)  t^{-2\alpha})  }\right) \ell d\ell \right)\\ 
\stackrel{(c)}{=} &\exp \left( -2 \pi \lambda_B\frac{-4}{\alpha} \int_{X_0}^{X_R}\left(1- e^{-(k_1(s)X  +k_2(s) X^2)  }\right) X^{-\frac{2}{\alpha}-1} dX \right)\\ 
\stackrel{(d)}{=} &\exp \left( 2 \pi \lambda_B\frac{4}{\alpha} \int_{X_0}^{X_R}\sum_{i=1}^{\infty}b_i(s) \; {X^{i-\frac{2}{\alpha}-1}} dX \right)\\ 
\stackrel{(e)}{=} &\exp \left( 2 \pi \lambda_B\frac{4}{\alpha} \sum_{i=1}^{\infty} \frac{ b_i(s)}{{i-\frac{2}{\alpha}}} \;\left( {X_R^{i-\frac{2}{\alpha}}}-{X_0^{i-\frac{2}{\alpha}}} \right)  \right),
\end{split}
\end{align} 
where (a)  follows from the LT of $Z_j$ where $Z_j$ is a Gaussian random variable with $\mu_{Z_j}$, and $\sigma^2_{Z_j}$ is given by {\bf Lemma~4}, (b) follows by  substituting $k_1(s)=\mu_{Z_j} s \; P/t_j^{-\alpha} $ and $k_2(s)=\frac{1}{2 t_j^{-2\alpha}} \sigma^2_{Z_j} s^2  P^2$ and then we apply PGFL w.r.t $\ell_{\rm j}$ where $t=\sqrt{\left(\frac{\ell}{2}\right)^2+(H_B-H_R)^2}$and $t_{\rm j}$ is given in \eqref{tj}. For simplicity, (c) is obtained by changing of variable $\ell$ to $X$, i.e., $X=\left({\frac{\ell^2}{4}+(H_B-H_R)^2}\right)^{-\alpha/2}$, where  $X_0=\left({H_B-H_R}\right)^{-\alpha}$ and $X_R=\left({\frac{{R^2}}{4}+(H_B-H_R)^2}\right)^{-\alpha/2}$. {Note that (d) is obtained by using Taylor's series expansion of $\exp(-k_1 x-k_2 x^2)$ and $b_i(s)$ denotes the coefficients of the expanded Taylor series. Finally, (e) is obtained by solving the integral.}
\end{proof}

		\begin{figure*}[t]
	\begin{minipage}{0.48\textwidth}
		\includegraphics[scale=0.60]{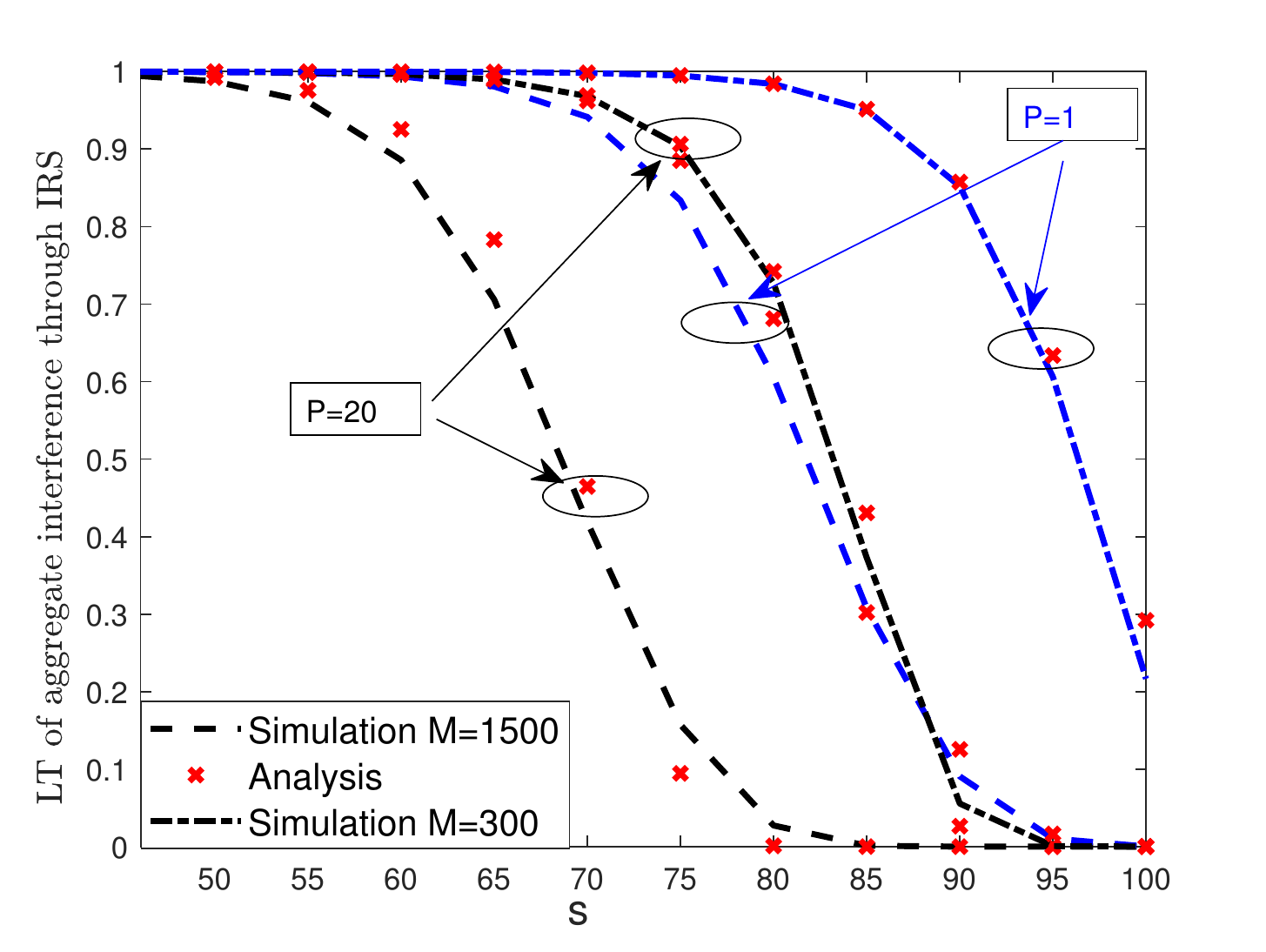}
	\caption{Conditional LT  of aggregate interference from IRSs (excluding the nearest IRS), $\mathcal{L}_{I_R}(s)$ in  \eqref{eq:LTInDirectDforID}, for $\lambda_R=2 \lambda_0, M=300$ and $\lambda_R=10 \lambda_0, M=1500$ with $P=1$ and $P=20$, using 
	Monte-Carlo simulations.}
	\label{fig:Figure3_MGFInDirectFinal1}	
		\end{minipage}\hfill
		\begin{minipage}{0.48\textwidth}
		\includegraphics[scale=0.60]{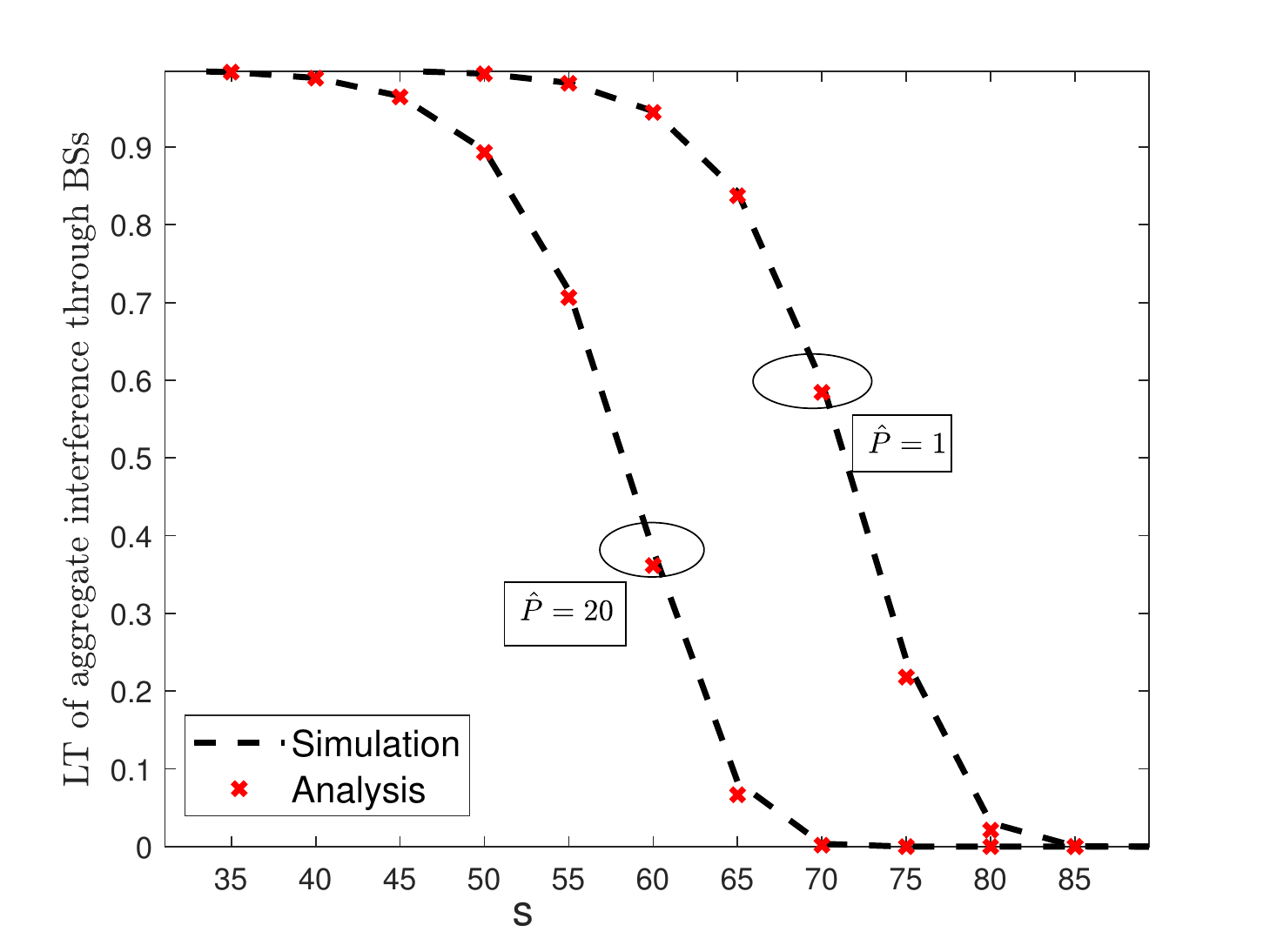}
		\caption{Conditional LT of aggregate interference from BSs (excluding the nearest BS in direct mode),  $\mathcal{L}_{I_B}(s)$ \eqref{eq:LTDirectforID}, for  $\hat{P}=P=1$, $\hat{P}=P=20$, and $\lambda_R=2 \lambda_0, M=300$, using Monte-Carlo simulations. }
		\label{fig:Figure2_MGFDirectFinal1}	 	
	\end{minipage}\hfill
	\end{figure*}

{
Fig.~\ref{fig:Figure3_MGFInDirectFinal1} validates the accuracy of LT of aggregate interference from  IRSs  for different number of IRSs, i.e., $M = 300$ and $M=1500$ and transmission power $P=1$~W and $P=20$~W. This figure  shows that, for a given value of $s$, increasing transmission power and IRS intensity decreases IRS interference. Clearly, the interference in higher power and higher intensity trend dominates compared to all other combinations of  power and IRS intensity. }{	
Similarly, Fig.~\ref{fig:Figure2_MGFDirectFinal1} validates the accuracy of  the LT of the aggregate interference from BSs (excluding the nearest BS) given in \eqref{eq:LTDirectforID}  as a function of  $s$. Again, the LT of aggregate interference  decreases with  increasing transmission power of BSs (i.e., the interference increases). Unlike 
Fig.~\ref{fig:Figure3_MGFInDirectFinal1},  neither the IRS intensity nor the total number of IRSs $M$ have any effect on $\mathcal{L}_{I_B}$ as the direct transmissions are independent of $\lambda_R$ or $M$.}


\section{Coverage Probability and Ergodic Capacity Characterization }
In this section, we first derive the coverage probability of an IRS-assisted  user and then the coverage probability of users who are supported by direct transmissions. Then, we derive the ergodic capacity and energy efficiency  of an IRS-assisted  user and the user supported by direct transmission from BS. Finally, we derive the overall network coverage, ergodic rate, and energy-efficiency considering the  fraction of IRS-assisted and direct users in the network.

\subsection{Coverage Probability (IRS-assisted Transmission)}
The coverage probability of the typical user  associated to nearest IRS in the IRS-assisted indirect mode of communication is defined as $ C_{\rm {ID}}=\Pr(\gamma_{ID}\ge \tau)$, where the SINR for  IRS-assisted indirect transmission is given as follows:
\begin{align}
    \gamma_{ID}=  \frac{S_{R_0}}
{I_{\rm B}+I_{R}+N_0}. \quad
      \end{align}
The coverage probability can be calculated numerically by using Gil-Paleaz inversion theorem  \cite{di2014stochastic} as shown in the following:
\begin{align}
\begin{split}
\label{eq:CID}
    C_{\mathrm{ID}}= &\Pr\left( \gamma_{ID} \ge \tau \right)\\=&\Pr\left({S_{R_0} (r_{0,0})}-\tau {I}_{R}\ge \tau \;  {I}_{B}+\tau \; N_0 \right)=\Pr\left(\Omega\ge \tau  {I}_{B}+\tau N_0 \right)\\
    =&\mathbb{E}_{r_{0,0}}\left[\frac{1}{2}-\frac{1}{\pi} \int_{0}^{\infty} \frac{\rm{Im}[\phi_{\Omega|r_{0,0}} (\omega) \mathcal{L}_{I_B}(- j \omega \tau ) e^{j \omega \tau  N_0 }]}{\omega} d\omega\right]\\
    =&\frac{1}{2}-\frac{1}{\pi} \int_{0}^{\infty} \frac{\rm{Im}[\phi_{\Omega} (\omega)\mathcal{L}_{I_B}(- j \omega \tau )e^{j \omega \tau  N_0 }]}{\omega} d\omega,
\end{split}
\end{align}
where 
\begin{align}
\label{OmegaRR}
\begin{split}
    \phi_{\Omega}(\omega)=&\mathbb{E}_{r_{0,0}}[\phi_{\Omega|r_{0,0}} (\omega)]=\mathbb{E}_{r_{0,0}}[e^{-j \omega \Omega}]=\mathbb{E}_{r_{0,0}}[ \mathcal{L}_{{S}_{R_0} |r_{0,0}}\mathcal{L}_{{I}_{R} |r_{0,0}}(-j \omega\tau )].
    \end{split}
\end{align}
Note that $\mathcal{L}_{I_B}(- j \omega \tau ) =\mathbb{E}_{d_0}[\mathcal{L}_{I_B|d_{0}}(- j \omega \tau )]$ is independent of $r_{0,0}$. Substituting \eqref{eq:LTInDirectDforID} and \eqref{lapsro} in \eqref{OmegaRR} and then substituting  \eqref{OmegaRR} and  \eqref{eq:LTDirectforID} in \eqref{eq:CID}, we obtain $C_{\rm ID}$. The distribution of the distance of the  nearest IRS at height $H_R$ to the typical user  is  given as follows \cite{srinivasa2009distance}:
\begin{align}
\label{eq:Dist}
     f_{r_0}(r_{0,0})=& \frac{2M r_{0,0}}{R^2} \left(1-\frac{r_{0,0}^2-H_R^2}{R^2}\right)^{M-1}.
\end{align}
Also,  the distance of the  nearest BS at height $H_B$  to the typical user is given as:
\begin{align}
\label{eq:Dist2}
        f_{d_0}(d_0)=&2\pi \lambda_{B} d_0 e^{-\pi \lambda_{B} (d_0^2-H_B^2)}.
\end{align}
\subsection{Coverage Probability (Direct Transmission)}
The  coverage probability of the typical user  with the direct mode of communication is defined as $ C_{\rm D}=\Pr(\gamma_{D}\ge \tau)$, where the SINR of the direct communication mode is given as \begin{align}
\label{eq:gamaaDirct}
     \gamma_{D}= \frac{S_{D_0}}{\hat{I}_{ B}+\hat I_{R}+N_0} .
     \end{align} Now by substituting \eqref{SigPowerDirect}, the coverage probability $C_{\rm D}$ can be written as follows:
\begin{align}
\label{eq:PcD}
\begin{split}
    C_{\rm {D}}= &\Pr\left({ \vert h_0\vert^2 }\ge d^{\alpha}_{0} \frac{\tau  \beta^{-2}  }{{\hat{P}}}\; \left( \hat I_{\rm B}+\hat I_{\rm R}+N_0 \right)\right)\\=& \mathbb{E}_{d_0}\left[e^{-  \frac{\tau \;d^{\alpha}_{0} N_0}{\beta^{2}\hat{P}} }\;\mathcal{L}_{\hat I_B}\left( \frac{\tau \;d^{\alpha}_{0}}{\beta^{2}\hat{P}} \right)\;\mathcal{L}_{ \hat I_R}\left( \frac{\tau \;d^{\alpha}_{0}}{\beta^{2}\hat{P}} \right)\right],
\end{split}
\end{align}
where
$\mathcal{L} (\cdot)$ is the LT and $d_0$ is the distance between the typical user and the nearest BS, i.e., $d_0=\sqrt{\ell_0^2+H_B^2}$. 
The distribution of the distance of the  nearest BS  is provided in \eqref{eq:Dist2}.

\begin{corollary}	
The LT of the aggregate interference to the typical user through all the BSs (except the associated BS$_0$) $ \mathcal{L}_{\hat{I}_B}$ in the the  direct  mode can then be  obtained as follows:
 \begin{align}
\label{eq:LTDirect}
\begin{split}
\mathcal{L}_{\hat I_B|{d_0}}(s )= \mathrm{exp}\left(- 2 \pi \lambda_B \frac{ d_{0}^{2-\alpha} s \hat{P} \beta^2  }{\alpha-2} \;_2 F_1\left(1,\frac{-2+\alpha}{\alpha};2-\frac{2}{\alpha}; -s\; \hat{P} \beta^2 d_{0}^{-\alpha} \right) \right), 
\end{split}
\end{align}
which is similar to \eqref{eq:LTDirectforID} with $P$ replaced with $\hat{P}$ for the direct mode.
    	\end{corollary}

\begin{corollary}
         \label{lammaLTDID}
   Similar to Lemma~\ref{lammaLTDIDforID}, the LT of  interference experienced by the typical user through the  all the IRSs and all the BSs  $\mathcal{L}_{\hat{I}_R}$ in direct communication mode  is given as follows:
    \begin{align}
\label{eq:LTInDirectDforD29}
\begin{split}
\mathcal{L}_{\hat{I}_{\rm R}}(s )\approx \exp \left( 2 \pi \lambda_B\frac{4}{\alpha} \sum_{i=1}^{\infty} \frac{ \hat{b}_i(s)}{{i-\frac{2}{\alpha}}} \;\left( {X_R^{i-\frac{2}{\alpha}}}-{X_0^{i-\frac{2}{\alpha}}} \right) \right),
\end{split}
\end{align}  
where $\hat{b}_i(s)$ denotes the Taylor's series expansion coefficients of $\exp(-\hat{k}_1(s) x-\hat{k}_2(s) x^2)$ and $\hat{k}_1(s)=\hat{\mu}_{Z_j} s \; \hat{P}/t_j^{-\alpha} $, and  $\hat{k}_2(s)=\frac{1}{2 t_j^{-2\alpha}} \hat{\sigma}^2_{Z_j} s^2  \hat{P}^2$,  $\hat{\mu}_{Z_j}= \mathbb{E}[ r_{{0,m}}^{-\alpha}]  (M {t^2_{j}})^{-\alpha/2}\; (1+\lambda) \quad \mathrm{and} \quad \hat{\sigma}^2_{Z_j}=2  \mathbb{V} [ r_{{0,m}}^{-\alpha}]  (M {t^2_{j}})^{-\alpha}(1+2\lambda)$.
\end{corollary}
Note that the difference arises from the fact that $\hat{I}_R$ has $M$ interfering IRSs in {\bf Corollary~4}, whereas in \textbf{Lemma~\ref{lammaLTDID}}, we have $M-1$ IRSs  contributing to the aggregate interference $I_R$.  
Finally, substituting \eqref{eq:LTDirect} and \eqref{eq:LTInDirectDforD29} in \eqref{eq:PcD}, we obtain the  coverage probability of direct link $C_{\rm D}$ conditioned on the distance {$d_0$}.


 \subsection{Ergodic Capacity}
   The achievable ergodic capacity of a typical user can be given by using the coverage probability expressions as shown below \cite{tabassum2018coverage}:
   $$
   \mathbb{E}[\log_2(1+\mathrm{SINR})]=\frac{1}{\rm ln(2)}\int_{0}^{\infty} \frac{P( \mathrm{SINR}>t)}{t}\; dt.$$ 
   However, the aforementioned evaluation adds one more layer of integration on top of the coverage probability. Therefore, we use an alternative LT-based approach to evaluate ergodic capacity by leveraging on Hamdi's lemma \cite{hamdi2010useful} given as follows:
\begin{align}
\mathbb{E}\left[\ln \left(1+\frac{X}{Y+N_0} \right)\right]=\int_{0}^{\infty}\frac{\mathcal{L}_Y(s)-\mathcal{L}_{X,Y}(s)}{s} \exp(-N_0 s) ds,
      \end{align}
      where $\mathcal{L}_Y(s)$ and $\mathcal{L}_{X,Y}(s)$ represent the LT of $Y$ and joint LT of $X$ and $Y$, respectively. 
      Subsequently, we derive  the ergodic capacity of the typical IRS-assisted user as follows: 
\begin{align}
\label{eq:rateIRS}\begin{split}
R_{ID}=\int_{0}^{\infty}\frac{\mathcal{L}_{{I}_B}(s)\mathcal{L}_{{I}_R}(s)-\mathcal{L}_{{I}_B}(s)\mathcal{L}_{{I}_R}(s)\mathcal{L}_{S_{R_0}}(s)}{s} \exp(-N_0 s) ds,
\end{split}
\end{align}
Similarly, the ergodic capacity of the typical user in direct mode $R_{\rm D}$ is given as follows:
\begin{align}
\label{eq:RateDirect}
\begin{split}
R_{D}=\int_{0}^{\infty}\frac{\mathcal{L}_{\hat{I}_B}(s)\mathcal{L}_{\hat{I}_R}(s)-\mathcal{L}_{\hat{I}_B}(s)\mathcal{L}_{\hat{I}_R}(s)\mathcal{L}_{S_{D_0}}(s)}{s} \exp(-N_0 s) ds,\\
\end{split}
      \end{align}
      where $\mathcal{L}_{S_{D_0}}(s)=\mathbb{E}[\mathcal{L}_{S_{D_0}|d_0}(s)]$ and  $\mathcal{L}_{S_{D_0}|d_0}(s)={ \frac{1}{1+s \; \hat{P} \beta^2 d_0^{-\alpha} } }$.

 
 \subsection{Energy Efficiency}
   We define the energy-efficiency of a typical user  by dividing the achievable rate with the network power consumption. The 
   energy-efficiency of IRS-assisted user is given as follows:  \begin{align}
   \label{eq:EEIRS}\begin{split}
\mathrm{EE}_{\rm ID}=\frac{\int_{0}^{\infty}\frac{\mathcal{L}_{{I}_B}(s)\mathcal{L}_{{I}_R}(s)-\mathcal{L}_{{I}_B}(s)\mathcal{L}_{{I}_R}(s)\mathcal{L}_{S_{R_0}}(s)}{s} \exp(-N_0 s) ds}{p_{\rm BS}+p_{\rm U}+ P+p_{\rm  IRS}},
\end{split}
\end{align}
which is obtained by dividing \eqref{eq:rateIRS} with $p_{\rm IRS}$.
Similarly, the energy-efficiency of a typical user in the direct communication mode $EE_{\rm D}$ can be given by diving \eqref{eq:RateDirect} with the power consumption in the direct mode $\hat P$ as follows:
 \begin{align}
   \label{eq:EEdirect} \begin{split}
\mathrm{EE}_{\rm D}=\frac{\int_{0}^{\infty}\frac{\mathcal{L}_{\hat{I}_B}(s)\mathcal{L}_{\hat{I}_R}(s)-\mathcal{L}_{\hat{I}_B}(s)\mathcal{L}_{\hat{I}_R}(s)\mathcal{L}_{S_{R_0}}(s)}{s} \exp(-N_0 s) ds}{p_{\rm BS}+p_{\rm U}+ \hat{P}}.
\end{split}
\end{align}

\subsection{Overall Network Coverage, Ergodic Capacity, and Energy Efficiency}
The overall coverage probability of the typical user
 is derived as follows: 
 \begin{align}
 \label{eq:FullCoverage}
  C= (1-\mathcal{A})C_{D}+ \mathcal{A}C_{ID},  \end{align}
  where $\mathcal{A}$ represents the fraction of users in the system performing indirect IRS-assisted transmission, while $(1-\mathcal{A})$ represents the fraction of users performing direct transmission. 
Similarly, the overall achievable rate and energy-efficiency of the typical user can be derived as follows: 
   $R=(1-\mathcal{A})R_{D}+ \mathcal{A}R_{ID},$
   and 
    $\mathrm{EE}=(1-\mathcal{A})\mathrm{EE}_{D}+ \mathcal{A}\mathrm{EE}_{ID},$
    respectively.

    The fraction of IRS-assisted and direct users can be perceived in many ways. For instance, it can be considered that the fraction of IRS-assisted users is proportional to the number of IRSs in the network. In this case, $\mathcal{A}$ can be defined as $\frac{\lambda_R}{\lambda_R+\lambda_B}$. As an example, if there are five BSs and five IRSs, then  $\mathcal{A}=0.5$ assuming that one IRS can at-most provide service to one-user at a time. On the other hand,  the fraction of IRS-assisted users can be considered proportional to the  blocking probability of nearest direct link (as IRS is only associated to BS if there is a blocked direct link). For instance, considering a  Boolean blockage model with the assumption that number of blockages follow Poisson distribution \cite{sayehvand2020interference}, the probability of direct transmission can be given as $\exp(-(\eta d_0+u))$, where $\eta$ and $u$ are defined on the basis of the shape of considered blockages \cite{bai2014analysis}. Subsequently, the probability of blockages can be written as $\mathcal{A}=1-\exp(-(\eta d_0+u))$.
     Considering blockage  the SINR of the direct mode in \eqref{eq:gamaaDirct} modifies as 
      $\gamma_{D}= \mathcal{A}\frac{S_{D_0}}{\hat{I}_{ B}+\hat I_{R}+N_0}$ that results in modified coverage probability  $C_{\rm D}$ in \eqref{eq:PcD} as $$
    C_{\rm {D}}=  \mathbb{E}_{d_0}\left[e^{-  \frac{\tau \;d^{\alpha}_{0} N_0}{\mathcal{A}\beta^{2}\hat{P}} }\;\mathcal{L}_{\hat I_B}\left( \frac{\tau \;d^{\alpha}_{0}}{\mathcal{A}\beta^{2}\hat{P}} \right)\;\mathcal{L}_{ \hat I_R}\left( \frac{\tau \;d^{\alpha}_{0}}{\mathcal{A}\beta^{2}\hat{P}} \right)\right].
$$

\section{Numerical Results and Discussion}

	In this section, we validate the accuracy of our derived expressions  and then obtain useful insights related to different interference scenarios,  the  total number of IRSs in the setup, number of IRS elements and transmission power for different communication modes.  
	
	\subsection{Simulation Parameters}
	
	Unless stated otherwise, the simulation parameters are listed herein. The heights of IRSs and BSs are set to $H_R=10$~m, and $H_B=20$~m, respectively. The coverage radius is $R=700$~m. The transmission power for IRS-assisted mode and direct mode is  $P=20$~W and $\hat{P}=20$~W, respectively. The static power consumption of BS and user is $p_{\rm BS}=40$~dBm and $p_{\rm U}=10$~dBm, respectively \cite{zhang2021trade}. The phase resolution power consumption for 6- bits $p_{r}(6)=78$~mW. The total number of IRS elements per IRS is $N=50$, BS intensity within the coverage area is $\lambda_B=10^{-4}$, and the total number of IRSs in the coverage area $M=1500 $ that corresponds to $\lambda_R=M/\pi R^2\approx 10\times \lambda_B$. Also,  $\lambda_R$ is IRS intensity, path-loss exponent is $\alpha=4$, threshold on SINR $\tau=-10$~dB, and noise power spectral density  is $N_0=10^{-10}$~W/Hz.
	
\subsection{Validation of Analysis}

	\begin{figure*}[t]
	\centering
	
	\includegraphics[scale=0.70]{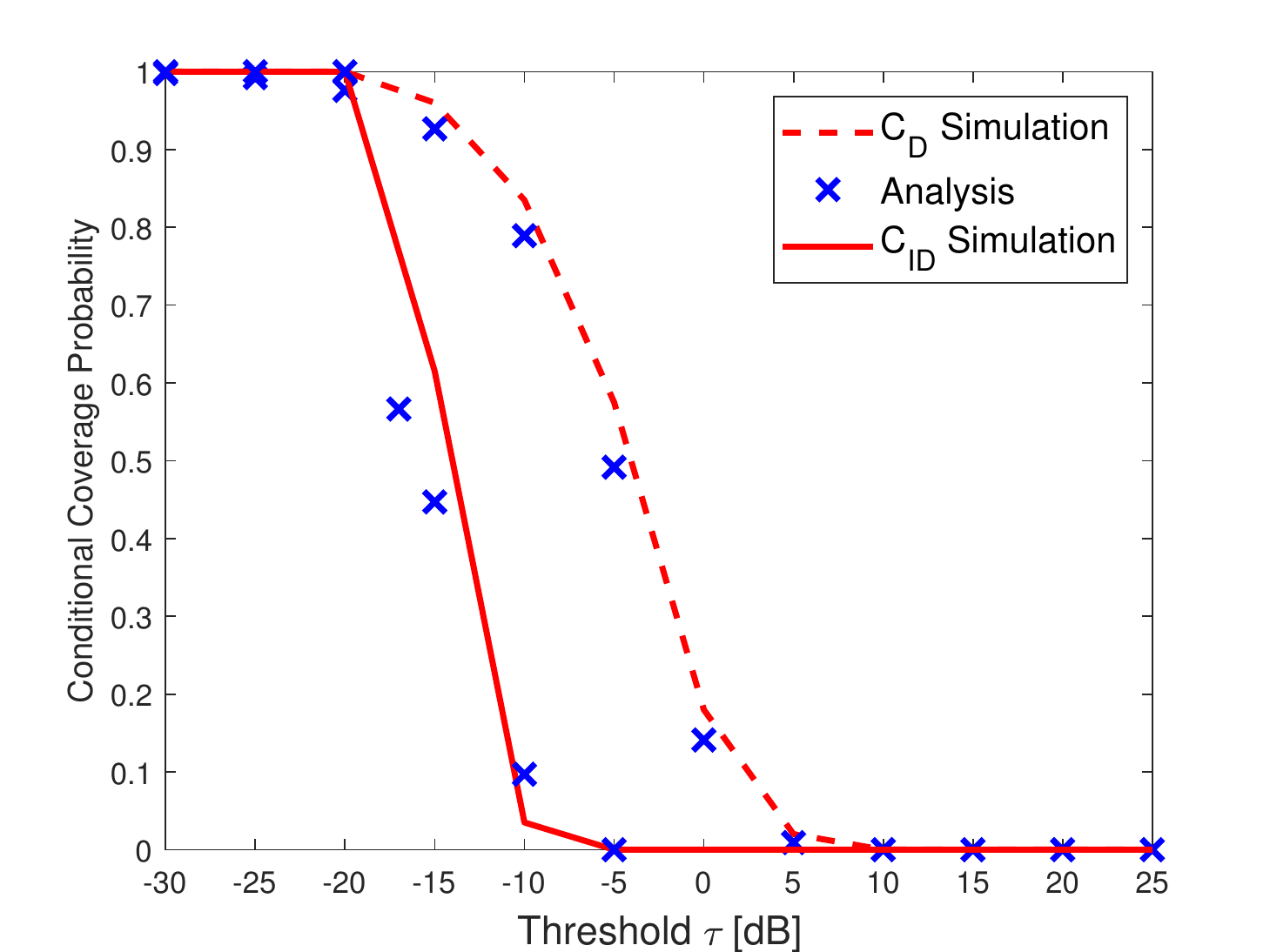}
	\caption{Validation of conditional coverage probability in  IRS-assisted and direct mode of communications derived in \eqref{eq:CID}  and \eqref{eq:PcD}, using Monte-Carlo simulations.}
	\label{fig:Figure_CoverageValidation}	
	\end{figure*}
	
	Fig.~\ref{fig:Figure_CoverageValidation}  compares the  coverage probability of IRS-assisted user and the user supported by the direct transmission as a function of the SINR threshold $\tau$  considering $P=\hat{P}=20$~W. Numerical results show that our  theoretical analysis and Monte-Carlo simulations match well. As expected, the conditional coverage probability decreases with the increase in SINR threshold for both types of users. Nevertheless, the coverage probability of IRS-assisted transmission lags behind the direct transmission even when the intensity of IRSs is higher than the intensity of BSs, i.e.,  $\lambda_R = 10 \lambda_B$. This fact signifies the efficacy of IRS deployments mostly in scenarios when the direct transmission link is blocked.
\subsection{Impact of BS Transmit Power on Direct Communication}	
	
	Fig.~\ref{fig:Figure_8_Rate}  compares the achievable data rate of IRS-assisted communication and the direct mode  considering  $\hat{P}=1$~W and $\hat{P}=5$~W. We observe that for smaller number of IRS elements, direct transmissions outperform the IRS-assisted transmissions.  As the number of IRS elements increases,  R$_{\rm ID}$ increases  because the IRS link gets stronger with more elements. An increase in IRS interference however degrades the achievable data rate  R$_{\rm D}$ in direct links. The figure also depicts that the performance of  IRS-assisted communication starts to exceed direct communication with lower IRS elements if the transmit power of BSs is low as can be seen from switching point $N=30$ and $N=60$ for $\hat{P}=1$ and $\hat{P}=5$, respectively. We note that, for a given deployment density of BSs and IRSs,  IRS-assisted mode is useful for a larger number of IRS elements and low transmit power of BSs in direct mode.   Evidently, a higher transmission power of direct user's BSs degrades IRS-assisted communication, which is opposite for direct communication.

	Similarly, Fig.~\ref{fig:Figure_8_EE} validates the accuracy of energy-efficiency considering $\hat{P}=1$~W and $\hat{P}=5$~W. As expected,  the IRS-assisted mode outperforms the direct mode for $N=40$ and $N=100$, for $\hat{P}=1$~W and $\hat{P}=5$~W, respectively.  Compared to  $\hat{P}=5$~W, energy-efficiency is lower for $\hat{P}=1$~W.
		\begin{figure*}[t]
	\begin{minipage}{0.48\textwidth}
		\includegraphics[scale=0.620]{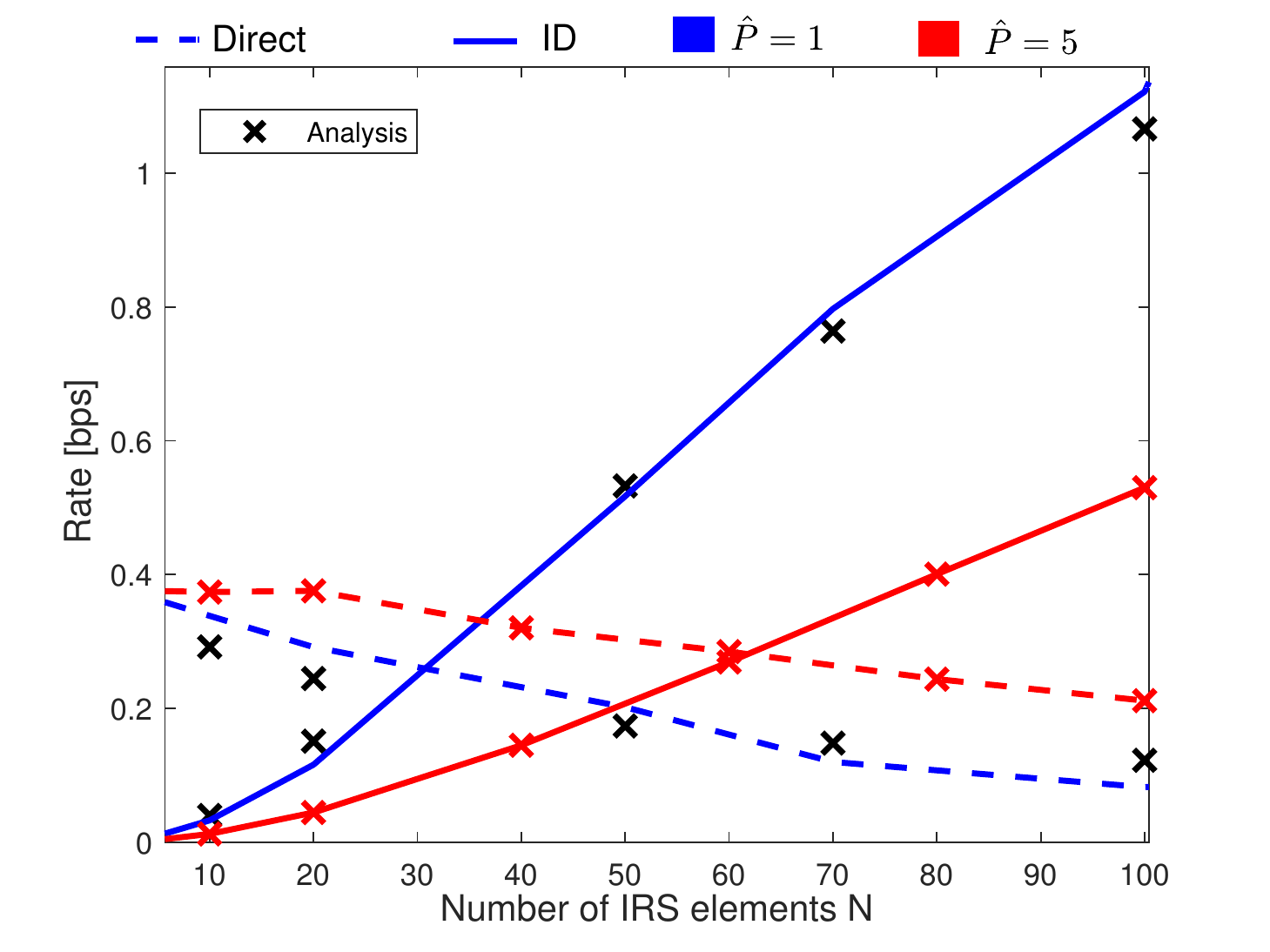}
	\caption{Analytical and simulation results on conditional achievable rate in  IRS-assisted and direct mode of communications derived in \eqref{eq:rateIRS} and \eqref{eq:RateDirect} with respect to IRS elements (for $\hat{P}=1$  and  $\hat{P}=5$).}
	\label{fig:Figure_8_Rate}
		\end{minipage}\hfill
		\begin{minipage}{0.48\textwidth}
			\includegraphics[scale=0.620]{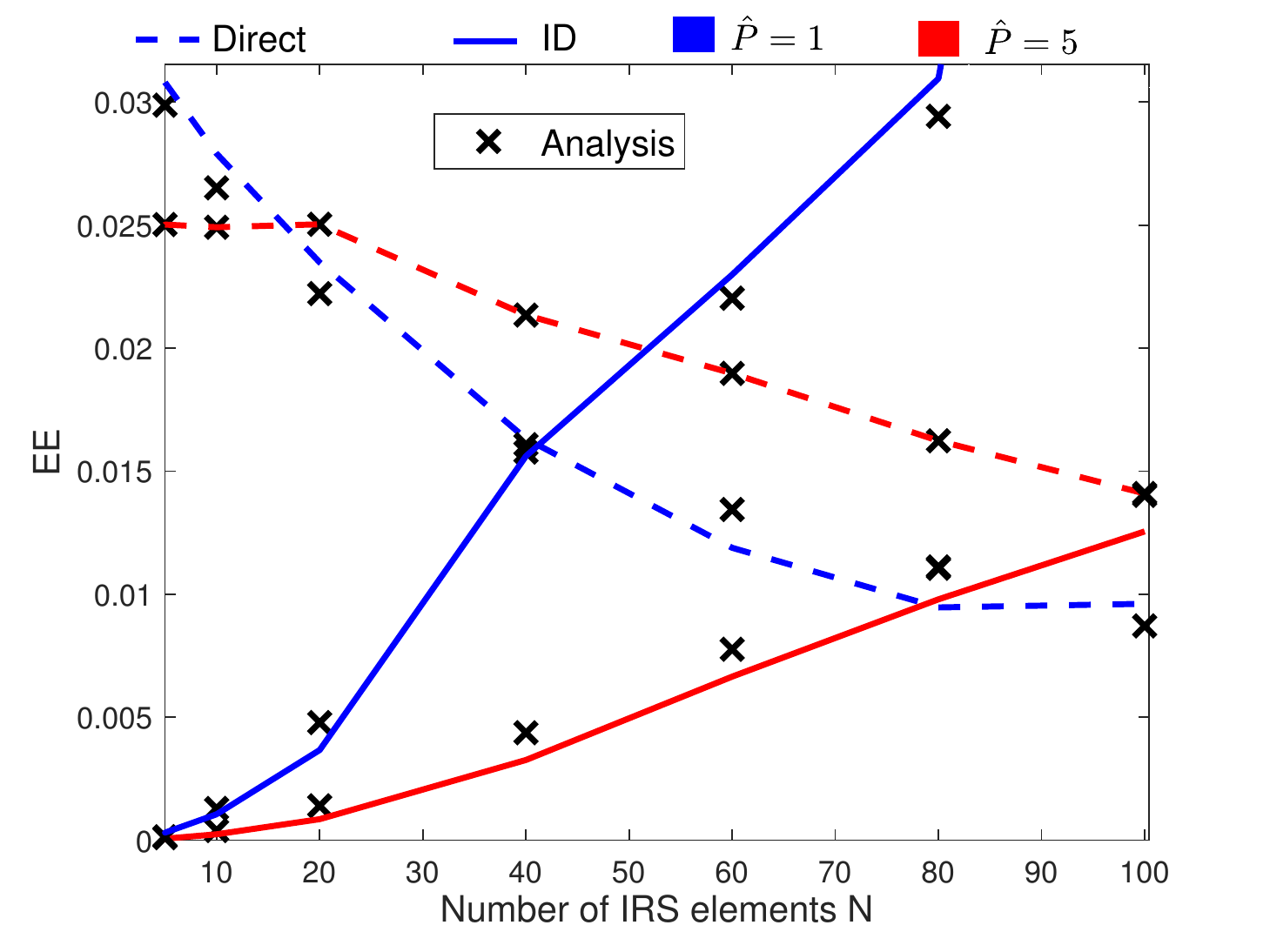}
	\caption{Validation of conditional EE for  IRS-assisted and direct mode of communications derived in \eqref{eq:EEIRS}and  \eqref{eq:EEdirect}, using Monte-Carlo simulations (for different number of IRS elements, $\hat{P}=1$  and  $\hat{P}=5$).}
	\label{fig:Figure_8_EE}	
	\end{minipage}\hfill
	\end{figure*}

\subsection{Impact of IRS Intensity on Direct and IRS-Assisted Communications}	
	\begin{figure*}[t]
	\begin{minipage}{0.48\textwidth}
			\includegraphics[scale=0.620]{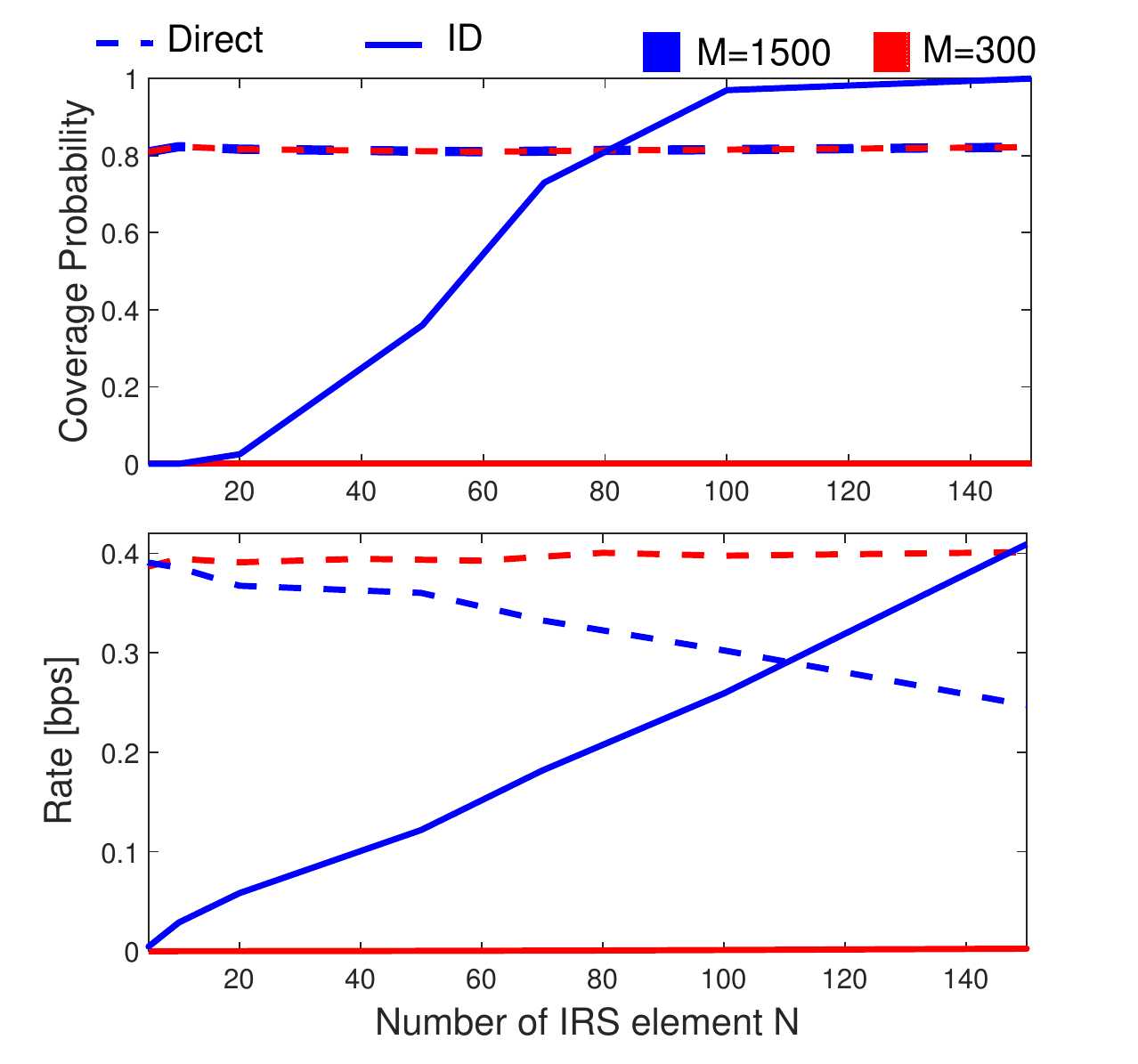}
		\caption{Comparison of conditional coverage probability and achievable rate for IRS-assisted mode and direct mode of communications with respect to number of IRS elements (for total number of IRSs $M=300$ and $M=1500$).
}
		\label{fig:Figure_CoverageSimCompare}	
		\end{minipage}\hfill
		\begin{minipage}{0.48\textwidth}
				\includegraphics[scale=0.620]{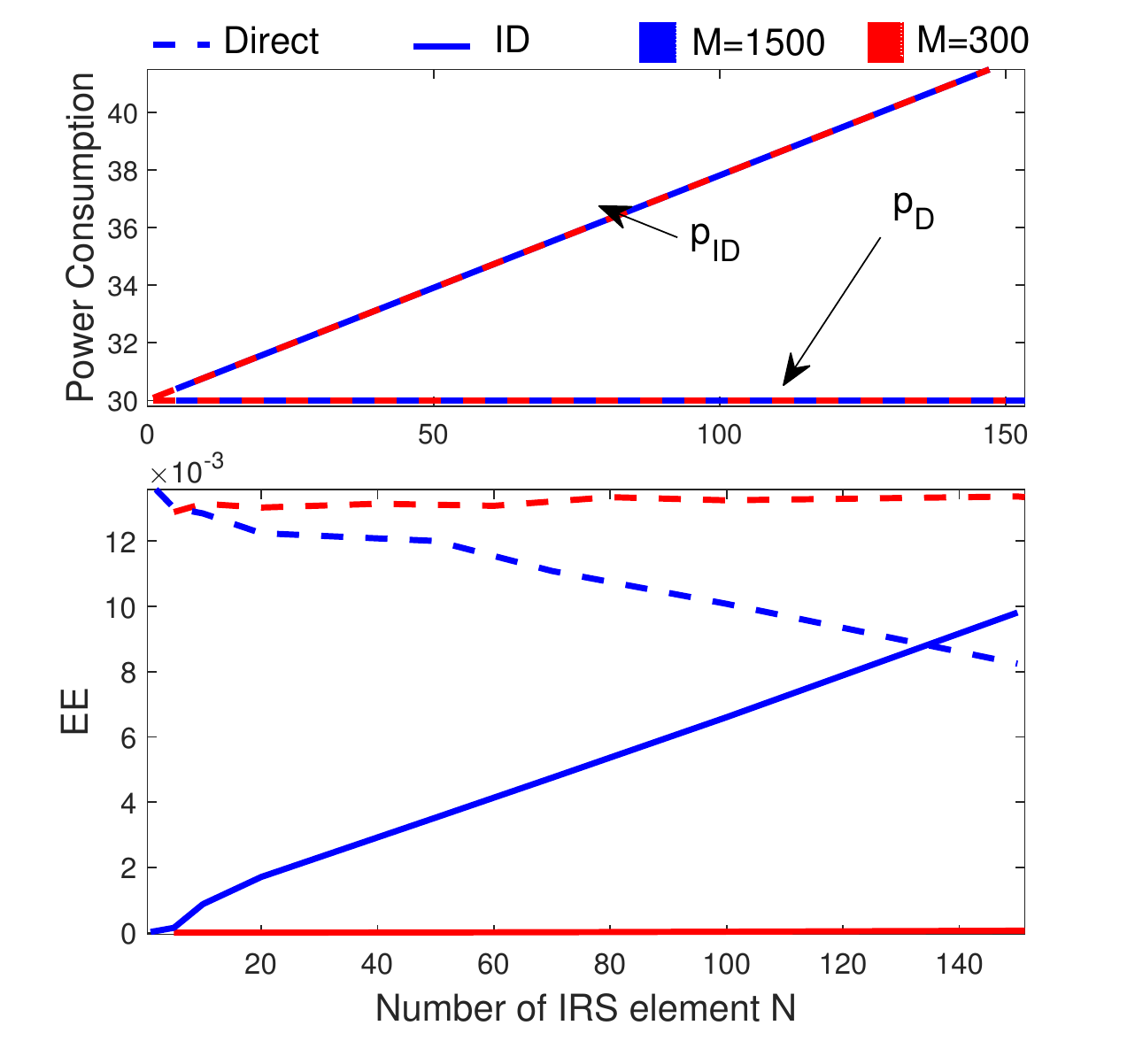}
	\caption{Comparison of power consumption and conditional EE for IRS-assisted mode and direct mode of communications with respect to number of IRS elements (for total number of IRSs $M=300$ and $M=1500$).
}
	\label{fig:Figure_EE_SimCompare}	
	\end{minipage}\hfill
	\end{figure*}

Fig.~\ref{fig:Figure_CoverageSimCompare} compares the coverage probability and rate for direct and IRS-assisted  communication as a function of the total number of IRS elements and IRSs with in the cell radius. We note that varying the number of IRS elements per IRS  have no significant impact on the coverage probability and rate for sparse deployment of IRSs $M=300$. However, the  coverage probability $C_{\rm ID}$ and achievable rate $R_{\rm ID}$ increases with the increase  in number of IRS elements for dense deployment of IRSs $M=1500$. This is encouraging as it shows that the impact of interference due to dense deployment of IRSs is not significant. On the other hand, the rate of the direct communication  decreases with the increasing IRS elements, especially for dense deployment of IRSs since the IRS interference becomes significantly dominant.  

Fig.~\ref{fig:Figure_EE_SimCompare} shows power consumption and EE for the IRS-assisted and the direct  modes of communication with respect to the number of IRSs  $M=300$ and $M=1500$. The figure presents that the $ p_{\rm ID}$ increases with the increase in $N$ as expected since $p_{\rm ID} \propto N$. However, the direct mode power consumption $p_{\rm D}$ remains same since   $p_{\rm D}$ is not the function of $N$. It is also clear that $M$ does not have any impact on the power consumption since $p_{\rm ID}$ is defined based on total system power consumption per user (refer to Section~\ref{subsec:PowerModel}) and a user is assumed to be connected with only one IRS at a time.  The  energy efficiency follows the same trend as conditional rate yet with the smaller slope due to the  increasing power of indirect mode that appears in the denominator of EE. 


	\begin{figure*}[t]
	\begin{minipage}{0.48\textwidth}
			\includegraphics[scale=0.630]{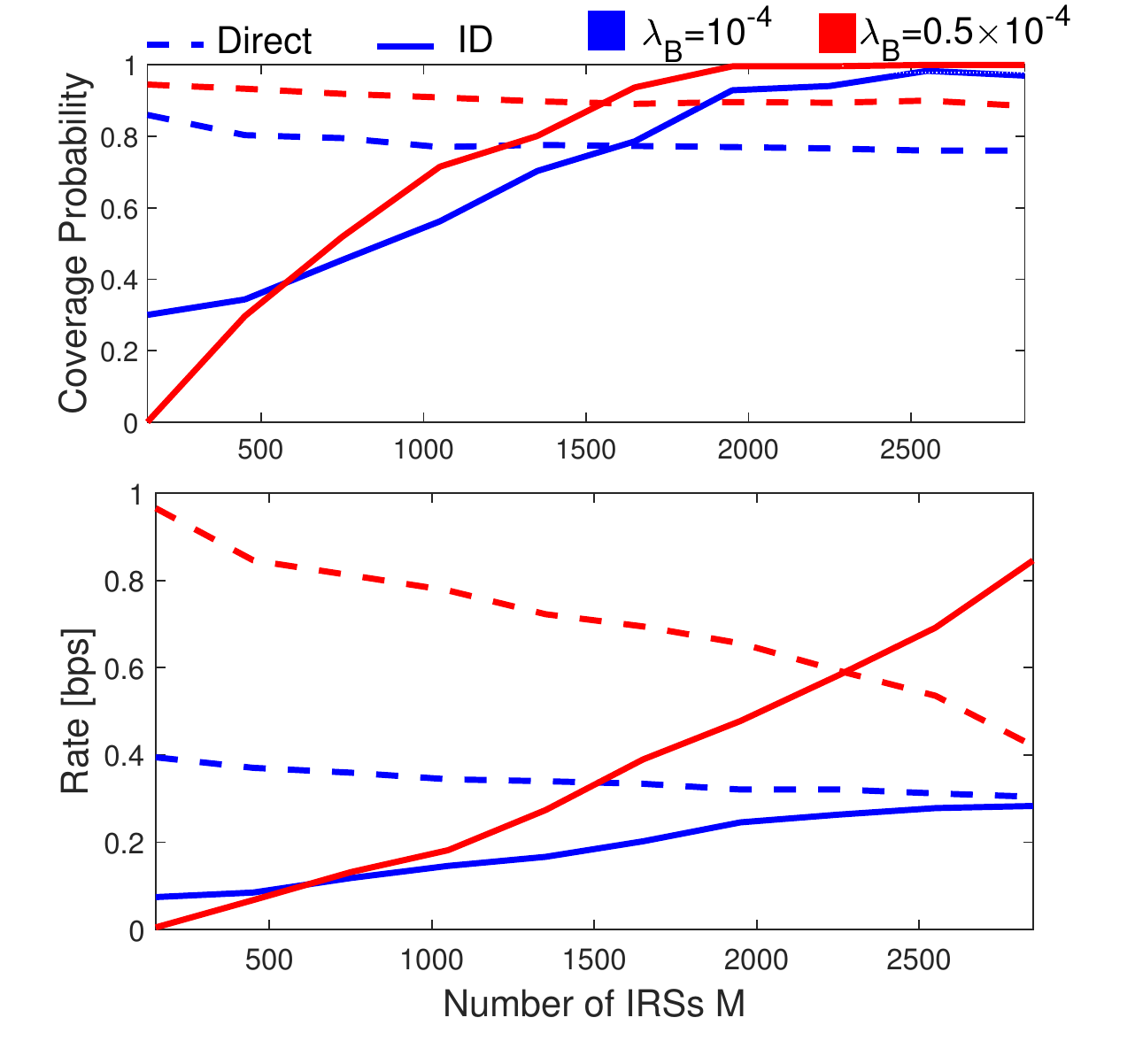}
		\caption{Comparison of conditional coverage probability and achievable rate  for IRS-assisted mode and direct mode of communications with respect to total number of IRSs  (for BS intensity  $\lambda_B=10^{-4}$ and $\lambda_B=0.5 \times 10^{-4}$, and $N=100$).
}
		\label{fig:FigureSetcoverageForIRS}	
		\end{minipage}\hfill
		\begin{minipage}{0.48\textwidth}
			\includegraphics[scale=0.630]{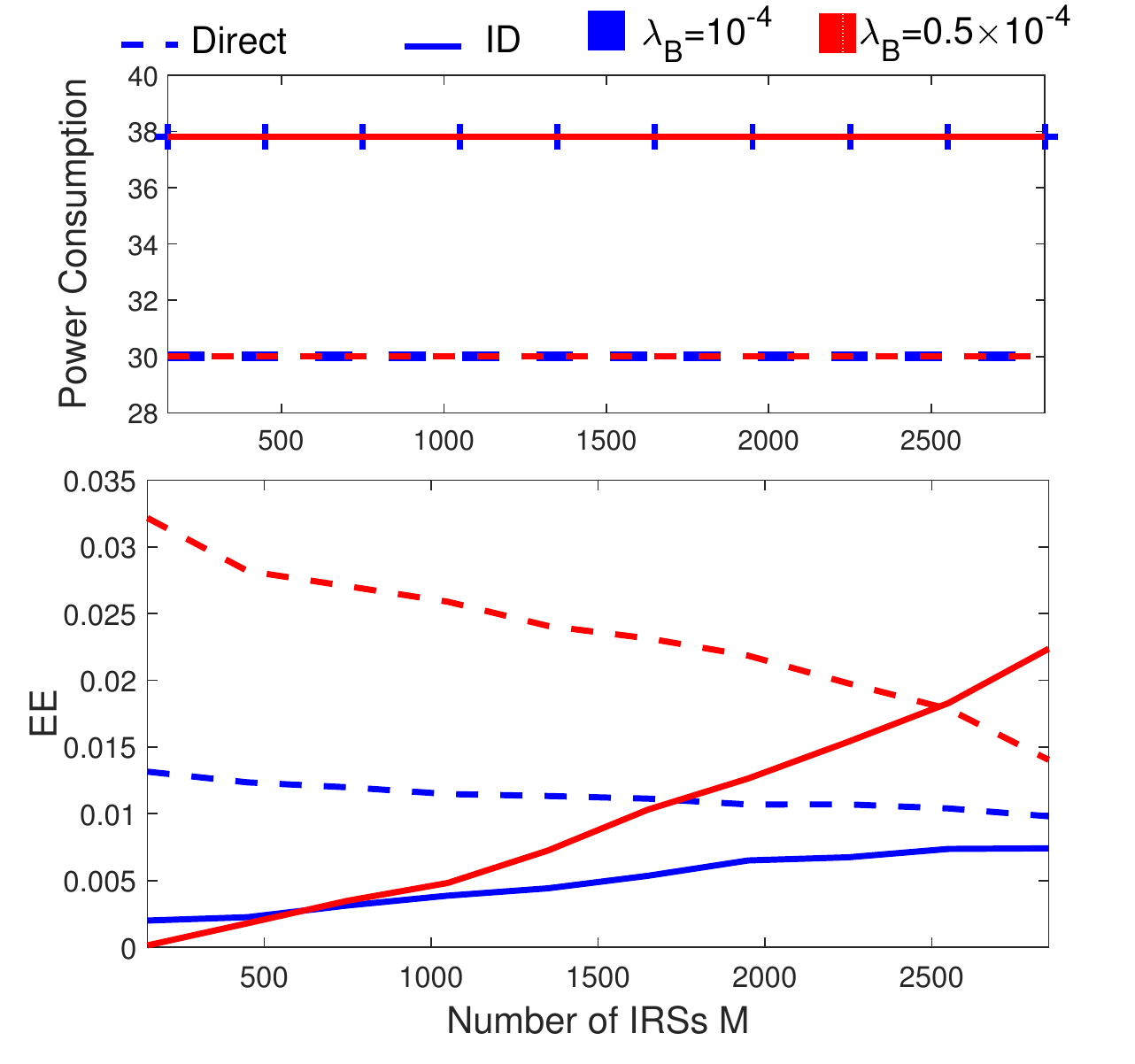}
	\caption{Comparison of power consumption and conditional EE  for IRS-assisted mode and direct mode  with respect to total number of IRSs (for BS intensity  $\lambda_B=10^{-4}$ and $\lambda_B=0.5 \times 10^{-4}$, and $N=100$).
}
	\label{fig:FigureSetEEForIRS}	
	\end{minipage}\hfill
	\end{figure*}
	
	\subsection{Impact of BS Intensity on Direct and IRS-Assisted Communications}	
Fig.~\ref{fig:FigureSetcoverageForIRS} compares the  coverage probability and ergodic capacity for 
IRS-assisted and direct communication with respect to total number of IRSs  in the coverage area for BS intensity  $\lambda_B=10^{-4}$ and $\lambda_B=0.5 \times 10^{-4}$. We observe that $C_{\rm ID}$ increases as total number of IRSs in the cell increases. Also, a very subtle decrease in  $C_{\rm D}$ is observed for both $\lambda_B=10^{-4}$ and $\lambda_B=0.5 \times 10^{-4}$. This is because, as M increases, the IRS density increases and  the nearest IRS becomes closer to the user that corresponds to smaller $r_{0,0}$ and higher IRS received signal power that leads to improvement in $C_{\rm ID}$. Also, an increases in $M$ increases the interference coming from  the IRSs for the direct user resulting in a slight decrease in $C_{\rm D}$. The figure also shows that a more sparse BS deployment leads to a smaller coverage probability of direct communication mode, and  indirect coverage $C_{\rm ID}$ outperforms direct mode coverage for  $M>1700$ for both the values of $\lambda_B$.  A similar trend can be observed for the  achievable rate. This implies that density of deployment of IRSs (i.e., sparse BS deployment or dense IRS deployment) plays a significant role in the performance of 
IRS-assisted mode.

Fig.~\ref{fig:FigureSetEEForIRS} presents results on power consumption and EE for the direct and indirect  modes. Fig.~\ref{fig:FigureSetEEForIRS} follows the same trend of achievable rate as in Fig.~\ref{fig:FigureSetcoverageForIRS} with the difference in the slope of EE$_{\rm ID}$. 

	\begin{figure*}[t]
	\begin{minipage}{0.48\textwidth}
		\includegraphics[scale=0.630]{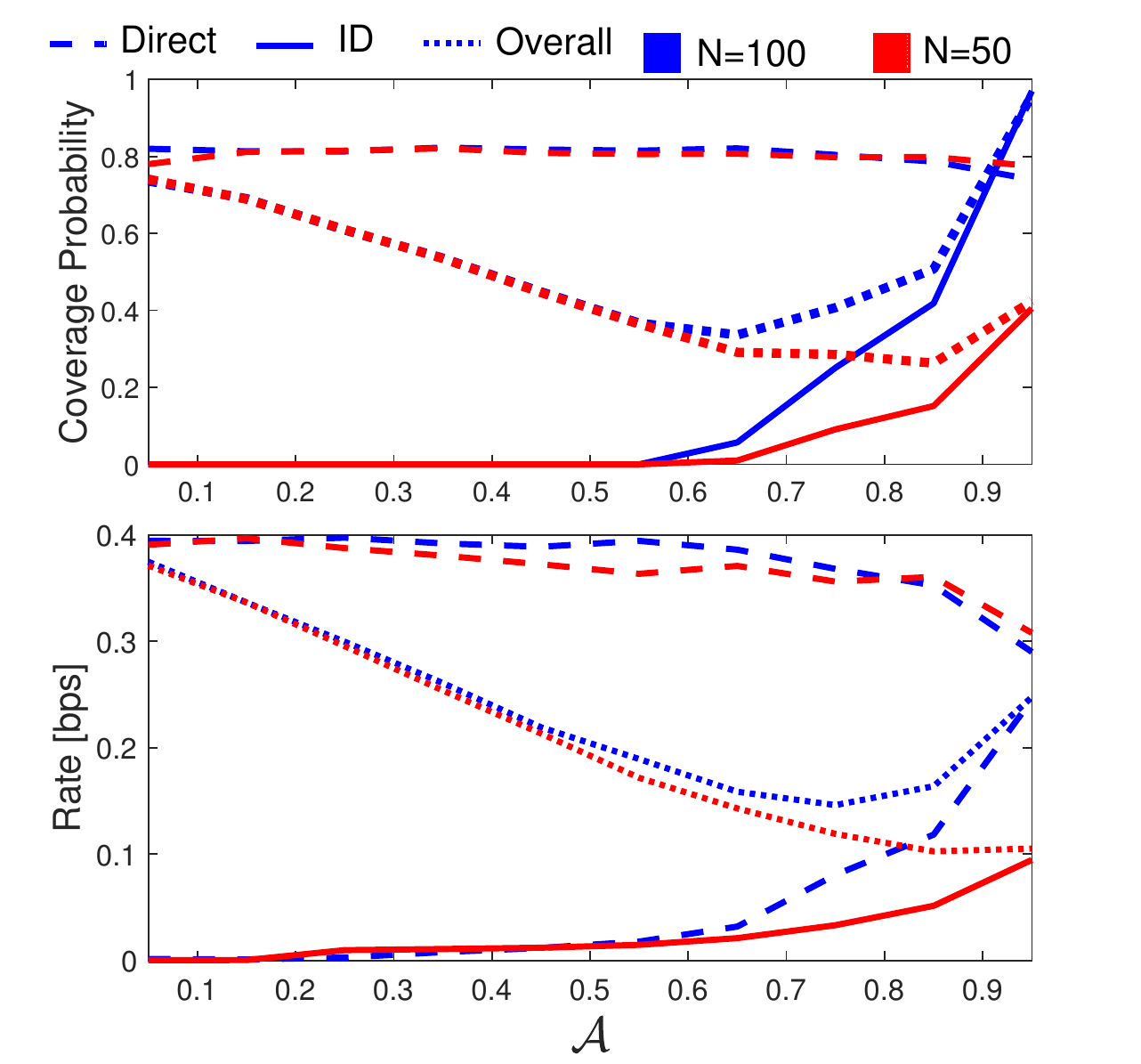}
		\caption{Comparison of conditional coverage probability and achievable rate  for IRS-assisted mode and direct mode, and overall performance with respect to the fraction of  users assisted by IRS $\mathcal{A}$ (for number of IRS elements  $N=50$ and $N=100$ per IRS surface).
}
		\label{fig:Figure_CoverageSimvsTransmissionModeFinal}	
		\end{minipage}\hfill
		\begin{minipage}{0.48\textwidth}
				\includegraphics[scale=0.630]{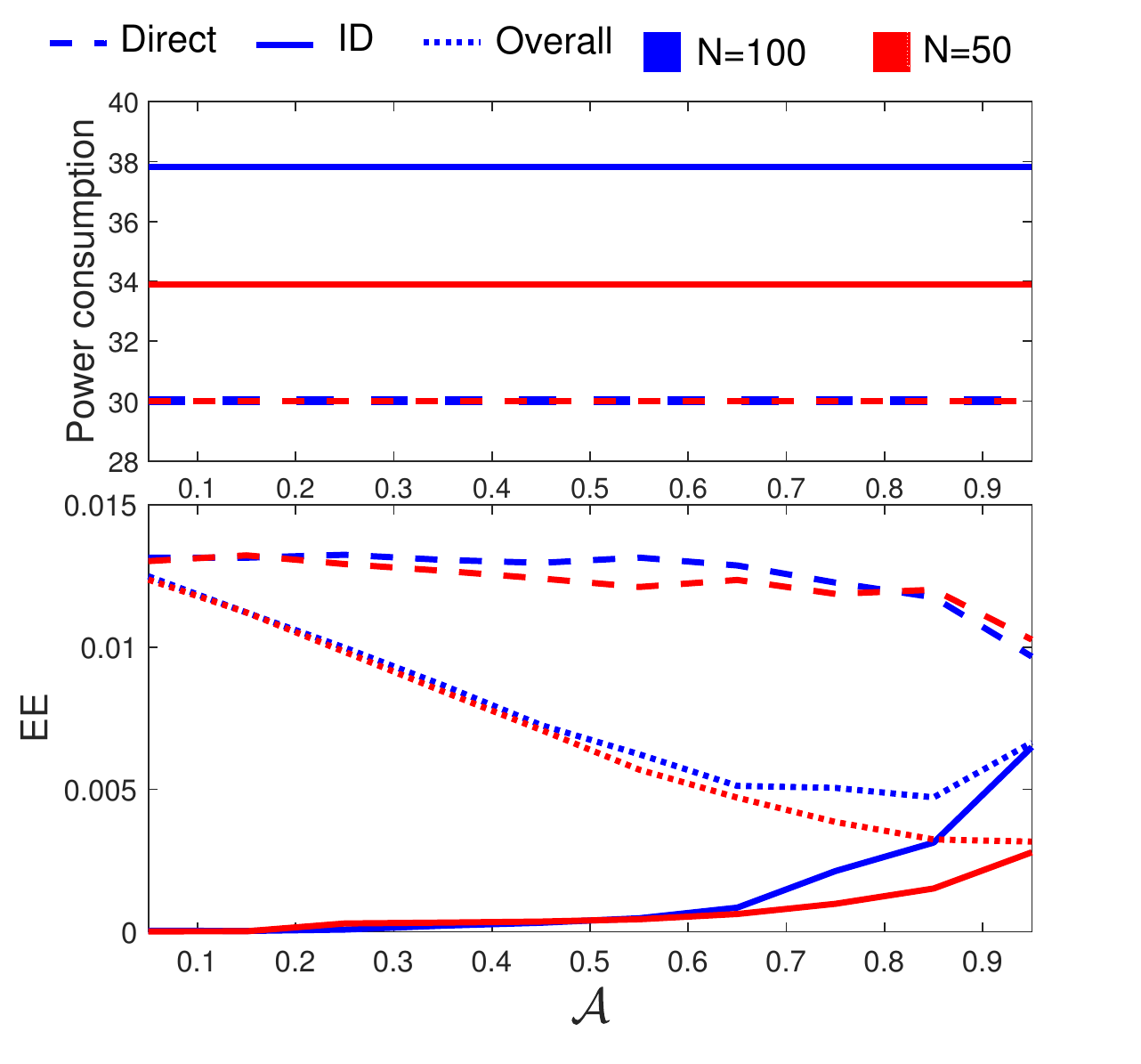}
	\caption{Comparison of power consumption, conditional energy-efficiency for IRS-assisted mode, direct mode, and overall EE  with respect to the fraction of  users assisted by IRS $\mathcal{A}$  (for number of IRS elements $N=50$ and $N=100$ per surface).
}
	\label{fig:Figure_EESimvsTransmissionModeFinal}	
	\end{minipage}\hfill
	\end{figure*}

Fig.~\ref{fig:Figure_CoverageSimvsTransmissionModeFinal} shows the impact of $\mathcal{A}$ on different  system performance measures. The coverage probability of IRS-assisted communication $C_{\rm ID}$ increases with $\mathcal{A}$ because this increases $\lambda_R=\frac{\mathcal{A}}{1-\mathcal{A}}\lambda_B$.  
The overall system coverage probability $P_C$ follows $C_{D}$ when $\mathcal{A}\approx 0$ which corresponds to very few or no IRS in the system. However, $P_C$ decreases up to $\mathcal{A}=0.6$ and then it starts to increase and converges to $C_{ID}$ when $\mathcal{A} \approx 1$ for $N=100$. Note that, for $\mathcal{A} \approx 0.95$, $\lambda_R=20 \lambda_B$. Moreover,  a decrease in  direct coverage probability  $C_D$ is also visible due to the aggregate interference coming from IRS. A similar trend is observed for $N=50$ with poorer  $C_{ID}$ than $C_{D}$ due to fewer IRS elements compared to the case when $N=100$. Also,
Fig.~\ref{fig:Figure_EESimvsTransmissionModeFinal} shows a similar 	trend in achievable rate because the power consumption does not change significantly.

\section{Conclusion}
	We have analyzed the downlink coverage probability, ergodic capacity, and energy-efficiency  performance for cellular networks under multi-BS and multi-IRS setup considering both the IRS-assisted communication and direct communication modes. We have observed  that using a larger number of IRS elements per IRS are crucial for IRS-assisted communication to outperform direct communication. Also, we have observed that   IRS-assisted communication becomes  dominant when   IRSs are densely deployed (i.e.,  when IRS intensity is larger than BS intensity). Also, for dense IRS deployment, the impact of IRS-interference  significantly decreases the performance of direct communication  and enhances IRS-assisted communication because the nearest IRS becomes closer to user.  Our results also have demonstrated the impact of fraction of indirect IRS-assisted users on the overall system performance and given insights on how to  select the proportion of direct or indirect IRS-assisted users in the network to achieve the desired trade-off between the degradation of direct communication and massive connectivity. The work can be extended to investigate the impact of multi-antennas at the BSs and the user devices. 
	




	\bibliographystyle{IEEEtran}
	\bibliography{IEEEabrv,Taniya_REF_2019_UAVIRS}




\end{document}